\documentclass[10pt,journal,final,twocolumn,twoside]{IEEEtran}
\usepackage[center]{subfigure}
\usepackage{amsmath,amsfonts,amssymb}
\usepackage[center]{caption}
\usepackage{graphicx,xcolor}
\usepackage{multirow,bigdelim,longtable}
\usepackage{amsthm}
\usepackage{hhline}
\let\labelindent\relax
\usepackage{enumitem}
\usepackage{flushend}

\usepackage{tikz,pgfplots}
\usetikzlibrary{intersections,calc,patterns,arrows}

\theoremstyle{definition}
\newtheorem{thm}{Theorem}
\newtheorem{lem}{Lemma}

\newtheorem{cor}{Corollary}
\usepackage{float}
\restylefloat{table}
\begin{document}

\title{When are dynamic relaying strategies necessary in half-duplex
  wireless networks?}

\author{Ritesh~Kolte,~\IEEEmembership{Student Member,~IEEE,} Ayfer~\"{O}zg\"{u}r,~\IEEEmembership{Member,~IEEE} and Suhas~Diggavi,~\IEEEmembership{Fellow,~IEEE}    \thanks{Manuscript received May 7, 2014; revised November 3, 2014; accepted January 14, 2015; date of current version January 29, 2015.  This work was presented in part at the IEEE International Symposium on Information Theory 2012 \cite{OzgDig} and the IEEE International Symposium on Information Theory 2013 \cite{KolOzg}. R.~Kolte and A.~\"{O}zg\"{u}r are with the Department of Electrical Engineering at
    Stanford University. S.~Diggavi is with the Department of Electrical Engineering at UCLA. (Emails: rkolte@stanford.edu,
    aozgur@stanford.edu, suhasdiggavi@ucla.edu). The work of R.~Kolte and A.~\"{O}zg\"{u}r was partly supported a Stanford Graduate Fellowship and NSF CAREER award 1254786. The research of S. Diggavi was supported in part by NSF grant 1314937 and a gift from Intel.}%
        \thanks{Communicated by A. S. Avestimehr, Associate Editor for Communications.}
        \thanks{Copyright (c) 2014 IEEE.}
}

\maketitle               

\begin{abstract} 
  In this paper we study a simple question: when are dynamic relaying
  strategies essential in optimizing the diversity-multiplexing
  tradeoff (DMT) in half-duplex wireless relay networks? This is
  motivated by apparently two contrasting results
  even for a simple $3$ node network, with a
  single half-duplex relay. When all channels in the system are
  assumed to be independent and identically fading, a \emph{static}
  schedule where the relay listens half the time and transmits half
  the time combined with quantize-map-and-forward (QMF) relaying
  is known to achieve the full-duplex
  performance. However, when there is no direct link
  between the source and the destination, a \emph{dynamic} decode-and-forward
  (DDF) strategy is needed to achieve the optimal
  tradeoff. In this case, a static schedule is
  strictly suboptimal and the optimal tradeoff is significantly worse
  than the full-duplex performance. In this paper we study the general
  case when the direct link is neither as strong as the other links
  nor fully non-existent, and identify regimes where dynamic schedules
  are necessary and those where static schedules are enough. We
  identify four qualitatively different regimes for the single relay channel
  where the tradeoff between diversity and multiplexing is significantly different. We
  show that in all these regimes one of the above two strategies is
  sufficient to achieve the optimal tradeoff by developing a new upper
  bound on the best achievable tradeoff under channel state
  information available only at the receivers. A natural next question
  is whether these two strategies are sufficient to achieve the DMT of
  more general half-duplex wireless networks with a larger number of
  relays. We propose a generalization of the two existing schemes
  through a dynamic QMF strategy, where the relay listens for a
  fraction of time depending on received CSI but not long enough to be
  able to decode. We show that such a dynamic QMF (DQMF) strategy
  is needed to achieve the optimal DMT in a parallel channel with two relays,
  outperforming both DDF and static QMF strategies.
\end{abstract}


\section{Introduction}
Diversity-multiplexing trade-off (DMT) \cite{ZheTse} captures the
inherent tension between rate and reliability over fading channels. It
has been used to demonstrate the value of relays in wireless networks
\cite{LanWor, LanTseWor, SenErkAaz, SenErkAaz2}.  The two critical
issues that complicate the problem in relay networks is who knows what
channel state and whether nodes can listen and transmit at the same
time ({\em i.e.,} half or full duplex).  The DMT of \emph{full-duplex}
(AWGN) wireless networks can be fully characterized, even with only
receiver channel knowledge, which can be forwarded to the
destination. For any statistics of the channel fading and any network
topology, it can be achieved by a quantize-map-and-forward (QMF)
strategy introduced in \cite{AveDigTse}. This is a simple consequence
of the fact that QMF achieves the capacity of wireless relay networks
within a constant gap without requiring (transmit) channel state
information (CSI) at the relays.

In current wireless systems, however, nodes operate in a half-duplex
mode, \emph{i.e.,} they can not simultaneously transmit and receive
signals on the same frequency band. Designing DMT optimal strategies
for half-duplex networks is more challenging as it also involves an
optimization over the listen and transmit schedules for the relays,
which could be dynamic, \emph{i.e.,} dependent on the received signals
and the channel state.  Another challenge in a fading environment
is that transmit CSI is typically unavailable at the nodes; this
necessitates the design of relay listen-transmit schedules that are either
static or depend only on local receive CSI for dynamic
schedules. These challenges have contributed to the difficulty in
characterizing the optimal DMT of general half-duplex relay networks,
which remains an open problem. Even in the
special cases where the DMT has been characterized, the understanding
of what necessitates dynamic schedules is incomplete.

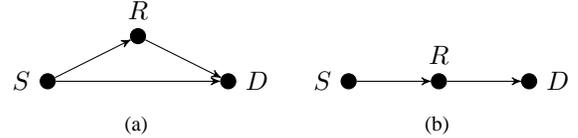
\begin{figure}
\centering
\subfigure[]{
\begin{tikzpicture}[scale=1.2]
\node (source) at (0,0) [circle,draw,fill,inner sep=0pt,minimum size=2mm,label=left:$S$] {};
\node (relay) at (1,0.5) [circle,draw,fill,inner sep=0pt,minimum size=2mm,label=above:$R$] {};
\node (dest) at (2,0) [circle,draw,fill,inner sep=0pt,minimum size=2mm,label=right:$D$] {};
\path[draw,->,>=stealth'] (source)  -- (relay);
\path[draw,->,>=stealth'] (relay)  -- (dest);
\path[draw,->,>=stealth'] (source)  -- (dest);
\end{tikzpicture}
\label{subfig:relay}}
\subfigure[]{
\begin{tikzpicture}[scale=1.2]
\node (source) at (0,0) [circle,draw,fill,inner sep=0pt,minimum size=2mm,label=left:$S$] {};
\node (relay) at (1,0) [circle,draw,fill,inner sep=0pt,minimum size=2mm,label=above:$R$] {};
\node (dest) at (2,0) [circle,draw,fill,inner sep=0pt,minimum size=2mm,label=right:$D$] {};
\path[draw,->,>=stealth'] (source)  -- (relay);
\path[draw,->,>=stealth'] (relay)  -- (dest);
\end{tikzpicture}
\label{subfig:line}}
\caption{(a) Single relay network, (b) Line relay network}
\end{figure}

Consider the simplest case where the communication between a source
and a destination is assisted by a single half-duplex relay. Two
settings for a single relay network have been considered and
characterized in the literature:
\begin{itemize}
\item When all links (source-destination, source-relay,
  relay-destination) are independent and identically fading (see
  Figure~\ref{subfig:relay}), \cite{PawAveTse} shows that the optimal
  DMT is achieved by the QMF scheme with a \emph{fixed RX-TX} schedule
  for the half-duplex relay that does not depend on the channel
  realizations. Here, the relay listens half of the total duration for
  communication, then quantizes and maps its received signal to a
  random codeword and transmits it in the second half. We call this
  strategy static QMF in the sequel. The performance meets the
  full-duplex DMT.
\item When there is no link between the source and the destination,
  the single relay channel of Figure~\ref{subfig:relay} reduces to the
  line topology in Figure~\ref{subfig:line}. In this case,
  \cite{GunKhoGolPoo} shows that the optimal DMT is achieved by a
  \emph{dynamic} decode-and-forward (DDF) strategy at the relay introduced in
  \cite{AzaGamSch}. In DDF, the relay listens until it gathers enough
  mutual information to decode the transmitted message so its RX time
  is dynamically determined as a function of its incoming channel
  realization and the targeted rate \cite{AzaGamSch}. The optimal
  performance does not reach the full-duplex DMT.
\end{itemize}
The two results suggest two different conclusions. While the first
result suggests that fixed schedules are sufficient to achieve the
optimal DMT with a half-duplex relay, even the full-duplex
performance, the second result establishes the necessity of dynamic
scheduling, which, even though  DMT optimal, does not meet the
full-duplex performance. In a practical setup, the source-destination
link can be expected to be neither as strong as the relay links (as in
the first setup in Figure~\ref{subfig:relay}) nor fully non-existent
(as in the second setup in Figure~\ref{subfig:line}). Given the
difference in the nature of the optimal strategies in the two
extremal cases, it is not clear which of these two strategies would be
optimal in a general setting where channel strengths are arbitrary; or
whether we need new strategies to achieve the optimal DMT in the general
case.

In this paper we answer these questions in the context of two
topologies: {\sf (i)} a relay channel where the direct link is neither
as strong as the other links nor fully non-existent, {\emph i.e.,}
where the different links scale differently. {\sf (ii)} a parallel
relay network which demonstrates the necessity for a new dynamic QMF
strategy. In the single relay channel, we demonstrate that the static
QMF scheme and DDF scheme are optimal in different regimes.  We
describe this in more detail below.

Let $(a,b,c)$ be the exponential orders of the average SNR's of the
source-relay (S-R), relay-destination (R-D) and source-destination
(S-D) channels respectively and $r$ be the desired multiplexing
rate. See Figure~\ref{subfig:abc_relay}. We show that:


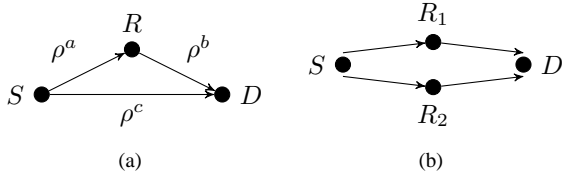
\begin{figure}
\centering
\subfigure[]{
\begin{tikzpicture}[scale=1.2]
\node (source) at (0,0) [circle,draw,fill,inner sep=0pt,minimum size=2mm,label=left:$S$] {};
\node (relay) at (1,0.5) [circle,draw,fill,inner sep=0pt,minimum size=2mm,label=above:$R$] {};
\node (dest) at (2,0) [circle,draw,fill,inner sep=0pt,minimum size=2mm,label=right:$D$] {};
\path[draw,->,>=stealth'] (source)  -- node[above left] {$\rho^a$} (relay);
\path[draw,->,>=stealth'] (relay)  -- node[above right] {$\rho^b$} (dest);
\path[draw,->,>=stealth'] (source)  -- node[below] {$\rho^c$} (dest);
\end{tikzpicture}
\label{subfig:abc_relay}}
\subfigure[]{
\begin{tikzpicture}[scale=1.2]
\node (source) at (0,0) [circle,draw,fill,inner sep=0pt,minimum size=2mm,label=left:$S$] {};
\node (relay1) at (1,0.25) [circle,draw,fill,inner sep=0pt,minimum size=2mm,label=above:$R_1$] {};
\node (relay2) at (1,-0.25) [circle,draw,fill,inner sep=0pt,minimum size=2mm,label=below:$R_2$] {};
\node (dest) at (2,0) [circle,draw,fill,inner sep=0pt,minimum size=2mm,label=right:$D$] {};
\node (out_source) at (0,0) [circle,inner sep=0pt,minimum size=3mm] {};
\node (out_dest) at (2,0) [circle,inner sep=0pt,minimum size=3mm] {};
\path[draw,->,>=stealth'] (out_source.north)  -- (relay1);
\path[draw,->,>=stealth'] (out_source.south)  -- (relay2);
\path[draw,->,>=stealth'] (relay1)  -- (out_dest.north);
\path[draw,->,>=stealth'] (relay2)  -- (out_dest.south);
\end{tikzpicture}
\label{subfig:parallel}}
\caption{(a) Single half-duplex relay channel, (b) Parallel
  half-duplex relay network}
\end{figure}
\begin{itemize}
\item when $c\geq\min(a,b)$, i.e. when the S-D link is as strong as or
  stronger than one of the relay links, static QMF achieves the
  full-duplex DMT. The result of \cite{PawAveTse} corresponds to the
  special case $(a,b,c)=(1,1,1)$.
\item when $c <\min(a,b)$, i.e. the S-D link is weaker than both the
  relay links but $r\leq c$, then the full-duplex DMT can still be
  achieved by static QMF.
\end{itemize}
The remaining regime is when $c <\min(a,b)$ and $r>c$, i.e. when the
direct link is weaker than one of the relay links but is not
sufficient alone to provide the desired multiplexing gain. To simplify
the analysis, we concentrate on the case where $a=b=p$.  We show that:
\begin{itemize}
\item when $r\geq p/2$, static QMF is again DMT optimal. It does not
  achieve the full-duplex DMT in this case, but it does achieve the
  best DMT under the more optimistic assumption that the TX-RX
  schedule can be optimized based on the knowledge of all
  instantaneous channel realizations in the network (i.e. global CSI
  at the relay). The result implies that this additional CSI is not
  needed. The largest achievable multiplexing gain is given by
  $\frac{p+c}{2}$.
\item when $c < r< p/2$, we show that DDF achieves the optimal DMT
  under local receive CSI. In this case, global CSI can improve the
  DMT. To the best of our knowledge, this is the first upper bound on
  the DMT trade-off under limited CSI. The result of
  \cite{GunKhoGolPoo} corresponds to the special case
  $(a,b,c)=(1,1,0)$, which falls in this regime.
\end{itemize}
These conclusions are summarized in Figure~\ref{fig:dmt_plot}.

The fact that DDF is optimal for small multiplexing gains and QMF with
a fixed schedule is optimal for large multiplexing gains is
qualitatively similar to the case when the relay is equipped with
multiple antennas but with identical fading for each link
\cite{KarVar,PraVar}. However, the two settings and the resultant
trade-offs are quite different. For example, here in the very low
multiplexing gain regime, QMF with a fixed schedule also becomes
optimal. Also in the intermediate multiplexing gain regime DDF is
optimal but cannot achieve the global CSI upper bound.  Moreover, for
this regime we needed to prove a new outer bound for the DMT under
(local) limited CSI. These ideas make the results distinct from the
multiple antennas case studied in \cite{KarVar,PraVar}.

The above discussion shows that static QMF and DDF are sufficient to
achieve the optimal trade-off in the single-relay channel when
$a=b$. Moreover, in all other scenarios studied in the literature such as \cite{YukErk, SreBirVij} where the optimal DMT is achieved by another strategy, it can be shown that either DDF or static QMF is also optimal with the added benefit of avoiding extra requirements on transmit
CSI at the relays. Therefore, the current results
in the literature for half-duplex relay networks (including our
result above) exhibit the following dichotomy: in all cases where the DMT of
half-duplex relay networks is known it is either achieved by DDF where
the relay waits until it can fully decode the source message, or by
QMF with a fixed schedule independent of the channel realizations. 
A natural question is whether these strategies are enough for half-duplex relay
networks.

We demonstrate that the answer is negative by developing a dynamic QMF
(DQMF) scheme where a relay listens for a fraction of time determined
by its receive CSI that is not necessarily long enough to allow
decoding of the transmitted message. The relay then quantizes, maps
and forwards the received signal as in the original QMF
\cite{AveDigTse}.  We show that for a specific configuration of two
parallel relays given in Figure~\ref{subfig:parallel}, DQMF is needed
to achieve the optimal DMT and it outperforms both DDF and static QMF.
We characterize the DMT and identify the optimal dynamic schedule at
the relays. This establishes the necessity of dynamic QMF for
achieving the DMT of general half-duplex relays. We also give
numerical evidence to show that DQMF might be needed for other regimes ($a\neq b$)
of the single-relay channel as well.

Our contributions are summarized as follows:
\begin{itemize}
\item demonstrating that DDF or static QMF are sufficient to get the
  optimal DMT when $a=b$ (see Figure \ref{subfig:abc_relay} and
  Theorem \ref{thm:mainres1}), and identify the corresponding regimes of the single relay channel.
\item a new outer bound for the DMT of the single-relay channel, when
  there is only local CSI (see Lemma \ref{lem:ddfreceivecsi}).
\item the necessity of a dynamic QMF strategy using a parallel relay
  network topology (see Theorem \ref{thm:mainres2}). We also give
  numerical evidence suggesting its importance in the single-relay
  channel when $a\neq b$ (see Section \ref{sec:anoteqb}).
\end{itemize}

The paper is organized as follows. In Section \ref{sec:model}, we
formulate the problem, establish the models and notation and describe some preliminaries used in the paper.  We state the main results in Section~\ref{sec:MainRes}.  We
prove the DMT optimality of DDF or static QMF for the single-relay channel
when $a=b$ in Section \ref{sec:relay}.  We briefly study the case
when $a\neq b$ in Section \ref{sec:anoteqb}. We give the proof for the
necessity of DQMF for parallel relay networks in Section
\ref{sec:parallel}. We end with a brief discussion in Section
\ref{sec:Conc}. Several of the proof details are given in the Appendices.



\begin{figure}
\centering
\input{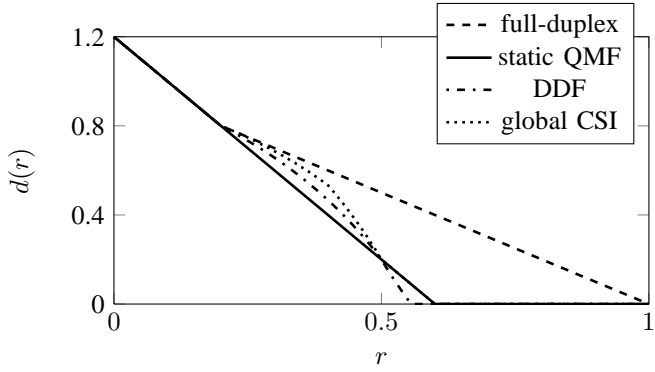}
\caption{DMT for $(a,b,c)=(1,1,0.2)$ where $(a,b,c)$ are the
  exponential orders of the average SNR's of the source-relay,
  relay-destination and source-destination channels respectively.}
\label{fig:dmt_plot}
\end{figure}

\section{Model and Preliminaries}\label{sec:model}
\subsection{Model}
We consider wireless networks where a source and a destination want
to communicate with the help of half-duplex relays. In this paper, we
consider the two configurations depicted in Figure
\ref{subfig:abc_relay} and \ref{subfig:parallel}. In the setup in Figure~\ref{subfig:abc_relay}, the source transmission is broadcasted to the relay and the destination, while the source and relay transmissions superpose at the destination. The relay is half-duplex and all nodes are equipped with a single
antenna. All channels are assumed to be flat-fading with
Rayleigh-distributed gains, i.e. the channel gains of the S-R, R-D and S-D
links are of the form $h_{sr}\rho^{a/2}$, $h_{rd}\rho^{b/2}$, $h_{sd}\rho^{c/2}$ respectively, where 
$h_{sr},h_{rd},h_{sd}$ are i.i.d. circularly-symmetric complex Gaussian random variables
$\mathcal{CN}(0,1)$ and $a,b,c\geq 0$. The additive noise
at every receiving node is assumed to be i.i.d. $\mathcal{CN}(0,1)$ and independent
across the nodes. Thus, $\rho^a$, $\rho^b$ and $\rho^c$ correspond to the average SNR's of S-R, R-D and S-D links and $a,b$ and $c$ are their exponential orders which can be different due to different path-loss and shadowing experienced by different links (see also \cite{NagPawTseNik}). We also define the exponential orders of the instantaneous SNR's for the three links as
$$\alpha\triangleq
\frac{\log(\left|h_{sr}\right|^2\rho^a)}{\log\rho},$$
\begin{equation}
\beta\triangleq
\frac{\log(\left|h_{rd}\right|^2\rho^b)}{\log\rho},
\label{eq:defalphabeta}
\end{equation}
$$\gamma\triangleq \frac{\log(\left|h_{sd}\right|^2\rho^c)}{\log\rho}.$$
These definitions are used extensively in the proofs in the following sections.

So, if $x_s[t]$, $x_r[t]$ denote the signals transmitted by the source and the relay respectively at time $t$, and $y_r[t]$ and $y_d[t]$ denote the signals received by the relay and the destination respectively, then the input-output relationships are given as follows:

\begin{IEEEeqnarray*}{l}
\text{If the relay is listening at time }t:\\
\quad\quad y_r[t] = h_{sr}[t]\rho^{a/2}x_s[t] + z_r[t]\\
\quad\quad y_d[t] = h_{sd}[t]\rho^{c/2}x_s[t] + z_d[t]\\
\\
\text{If the relay is transmitting at time } t:\\
\quad\quad y_r[t] = 0\\
\quad\quad y_d[t] = h_{sd}[t]\rho^{c/2}x_s[t] + h_{rd}[t]\rho^{b/2}x_r[t] + z_d[t],
\end{IEEEeqnarray*}
where $z_r[t]$ and $z_d[t]$ denote the additive noise at the relay and the destination respectively and the transmit signals are subject to a unit power constraint.
We assume quasi-static fading, i.e. the channel gains remain constant over the duration of the codeword and change
independently from one codeword to another. Local channel realizations
are known at the receivers but not at the transmitters, i.e. the relay
can track the realization of the source-relay link (and communicate it
to the destination) but can not track the relay-destination link. Similarly, the source node is not aware of the realizations of its outgoing channels. We also assume
that the codeword lengths are large enough so that an error occurs
only when the channel is in outage.

A sequence of codes $\{\mathcal{C}(\rho^a, \rho^b, \rho^c)\}$ indexed by $\rho$ with rate $R(\rho^a, \rho^b, \rho^c)$ and
average error probability $P_e(\rho^a, \rho^b, \rho^c)$ for a given
$(a,b,c)$ is said to achieve a multiplexing gain $r$ and diversity
gain $d$
if 
\begin{align}
\label{eq:dmt_def}
  \begin{split}
    \lim_{\rho\rightarrow\infty}\frac{R(\rho^a,
    \rho^b, \rho^c)}{\log\rho} & = r, \\
    \lim_{\rho\rightarrow\infty}\frac{\log P_e(\rho^a, \rho^b,
    \rho^c)}{\log\rho} & = -d.
  \end{split}
\end{align}
 For each multiplexing gain
$r$, the supremum $d(r)$ of diversity gains achievable over all
families of codes is called the diversity-multiplexing tradeoff (DMT)
of the half-duplex $(a,b,c)$-relay channel and is denoted by
$d_{(a,b,c)}^*(r)$.

In the setup of Figure~\ref{subfig:parallel}, the source communicates to the destination through two half-duplex parallel relays R$_1$ and R$_2$. By \emph{parallel}, we mean that there is no broadcasting from the source and no superposition at the destination and all links are independent of each other. This setup is different from the diamond network which has a similar topology. In the diamond network, the first hop resembles a Gaussian broadcast channel and the second hop resembles a Gaussian multiple-access channel, whereas the parallel relay setup we consider is composed of four orthogonal point-to-point channels. As before, nodes only know their incoming (receive) channel states and not the outgoing channel states. Here, we only focus on the case where all channels have the same average SNR $\rho$, which turns out to be sufficient for demonstrating the necessity of a dynamic QMF strategy, i.e. the channel gains of the S-R$_1$, S-R$_2$, R$_1$-D and R$_2$-D links are of the form $h_{sr_1}\rho^{1/2}$, $h_{sr_2}\rho^{1/2}$, $h_{r_1d}\rho^{1/2}$ and $h_{r_2d}\rho^{1/2}$ respectively where  $h_{sr_1},h_{sr_2},h_{r_1d},h_{r_2d}$ are i.i.d. circularly-symmetric complex Gaussian random variables $\mathcal{CN}(0,1)$. Thus, the average SNR of each of the four links is equal to $\rho$. If $x_{s_1}[t]$, $x_{s_2}[t]$ denote the signals transmitted by the source to the two relays respectively and $x_{r_1}[t]$, $x_{r_2}[t]$ denote the signals transmitted by the two relays at time $t$, then the received signals $y_{r_1}[t]$, $y_{r_2}[t]$ by the relays and $y_{d_1}[t]$, $y_{d_2}[t]$ by the destination are given as follows:

\begin{IEEEeqnarray*}{l}
\text{If relay } i \text{ is listening at time } t:\\
\quad\quad y_{r_i}[t] = h_{sr_i}[t]\rho^{1/2}x_{s_i}[t] + z_{r_i}[t]\\
\quad\quad y_{d_i}[t] = 0\\
\\
\text{If relay } i \text{ is transmitting at time } t:\\
\quad\quad y_{r_i}[t] = 0\\
\quad\quad y_{d_i}[t] = h_{r_id}[t]\rho^{1/2}x_{r_i}[t] + z_{d_i}[t],
\end{IEEEeqnarray*}

for $i=1,2 $, where $z_{r_1}[t]$, $z_{r_2}[t]$ and $z_{d_1}[t]$, $z_{d_2}[t]$ denote the additive noise at the two relays and the destination respectively and transmit signals are again subject to a unit power constraint. As in the previous setup, we assume quasi-static fading and sufficiently large codeword lengths. For this setup, we are interested in a sequence of codes $\mathcal{C}(\rho)$ for this setup indexed by $\rho$ with rate $R(\rho)$ and probability of error $P_e(\rho)$ achieving a multiplexing gain $r$ and diversity gain $d$ defined analogously to \eqref{eq:dmt_def}. For each multiplexing gain $r$, the supremum $d(r)$ of diversity gains achievable over all families of codes is called the diversity-multiplexing tradeoff (DMT) of the parallel relay network and is denoted by $d^*(r)$.

\subsection{Preliminaries}
In this subsection, we describe some results on capacity approximation that will be used in the proofs. In \cite{AveDigTse}, it was shown that for a Gaussian relay network, a quantize-map-forward (QMF) strategy can achieve rates that are within a constant gap of the capacity of the network. We specialize these results in this subsection to the setup in Figure~\ref{subfig:abc_relay}. Analogous results also apply for the setup in Figure~\ref{subfig:parallel}, and they are provided wherever required in Section~\ref{sec:parallel}.

First, consider the case when the relay in Figure~\ref{subfig:abc_relay} is full-duplex and the channel realizations $h_{sr}$, $h_{rd}$ and $h_{sd}$ are known to all the nodes. An upper-bound $C_u(\rho^a, \rho^b, \rho^c)$ on the capacity of this network is given by 
\begin{IEEEeqnarray}{lCl}
  C_u(\rho^a, \rho^b, \rho^c) \nonumber\\
  \quad\quad = \min\left\{\log\left(1+|h_{sr}|^2\rho^a + |h_{sd}|^2\rho^c\right)\right.,\nonumber\\
  \quad\quad\qquad\quad \left.\log\left(1+(|h_{rd}|\rho^{b/2} +
    |h_{sd}|\rho^{c/2})^2\right)\right\},\label{eq:fdcutset}
\end{IEEEeqnarray}
which is obtained by relaxing the standard information-theoretic cutset bound for this channel by exchanging the maximization over all joint input distributions and minimization over all cuts. 
The QMF scheme from \cite{AveDigTse} can achieve all rates upto \begin{equation}\label{eq:fdqmf}
C_u(\rho^a, \rho^b, \rho^c) - \kappa,\end{equation} where $\kappa$ is a
constant that is independent of the SNR and the channel
realizations. No transmit CSI is required by this scheme and therefore
can be applied as it is in our current outage setting.

Consider now the case of a half-duplex relay. In this case, the cutset bound and the QMF achievable rate depend on what channel state information (CSI) is assumed at the relay. In each case however, we have a constant gap result similar to the full-duplex case. 

If the half-duplex relay is assumed to follow a \emph{fixed} (non-random) listen-transmit schedule, that is independent of the channel realizations in the network, where it listens for a fixed $t$  fraction of the total time and transmits in the remaining $1-t$ fraction, then \eqref{eq:hdcutset} given at the top of the next page gives an upper bound on the achievable rate. 
\begin{figure*}[!th]
\normalsize
\begin{IEEEeqnarray}{l}
C_{h.d.}(\rho^a, \rho^b, \rho^c)\triangleq 
\min\left\{ \begin{array}{c} t\log(1+|h_{sr}|^2\rho^a +|h_{sd}|^2\rho^c)+(1-t)\log(1+|h_{sd}|^2\rho^c),\\
t\log(1+|h_{sd}|^2\rho^c)+(1-t)\log(1+ (|h_{rd}|\rho^{b/2} +|h_{sd}|\rho^{c/2})^2 )\end{array}\right\}\label{eq:hdcutset}
\end{IEEEeqnarray}
\hrulefill
\end{figure*}

The rate $R_{QMF}$ achieved by QMF with a {\em fixed} listen-transmit schedule where the relay listens for a fraction $t$ of the time is lower bounded in \cite[Section~VIII-C]{AveDigTse} as 
\begin{equation} \label{eq:hdqmf}
R_{QMF}\geq C_{h.d}(\rho^a, \rho^b, \rho^c)-\kappa,
\end{equation} where $\kappa$ as before denotes a constant independent of SNR and the channel realizations. Maximizing over all choices of $t$ yields the best rate achievable by QMF with such fixed transmit-receive schedules. 

In general, the listen-transmit schedule of the relay can be \emph{random}, in which case information can be transmitted through the sequence of listen-transmit states, and/or \emph{dynamic}, in which case it can depend on the instantaneous realizations of the channel coefficients. Consider the case when the relay has global CSI, i.e. it knows all the instantaneous realizations $h_{sd}, h_{sr}$ and $h_{rd}$ of the channels in the network. 
An upper bound on the capacity of the half-duplex $(a,b,c)$-relay channel with global CSI is given by \cite[Section VI]{OzgDig2} as $$\max_{{t(h_{sd},h_{sr},h_{sd})}}C_{h.d.}(\rho^a, \rho^b, \rho^c)+G$$ where $G$ is a constant independent of SNR and channel realizations and $C_{h.d.}$ is given in \eqref{eq:hdcutset}. Here $t$ again is the fraction of time relay listens to the source transmission but is now allowed to be a function of  $h_{sd}, h_{sr}$ and $h_{rd}$. It is also shown in \cite{OzgDig2} that a fixed dynamic QMF scheme with the same choice for $t$ that maximizes $C_{h.d.}(\rho^a, \rho^b, \rho^c)$ achieves rates that are within a constant gap of the above upper-bound, i.e. all rates less than $\max_{{t(h_{sd},h_{sr},h_{sd})}}C_{h.d.}(\rho^a, \rho^b, \rho^c)-\kappa$ are achievable by a fixed dynamic QMF scheme utilizing global CSI.

Finally, when the relay only has receive CSI (CSIR), i.e. it only knows the channel realization $h_{sr}$, an upper-bound and achievable lower-bound on the capacity of the relay channel can be obtained by adapting the proof in \cite[Section VI]{OzgDig2} to the case of limited CSI. Now, the choice of the listening time $t$ can only be a function of $h_{sr}$. Hence, we get the upper-bound on the capacity as $\max_{t(h_{sr})}C_{h.d.}(\rho^a, \rho^b, \rho^c)+G$ and lower-bound achievable by QMF as $\max_{t(h_{sr})}C_{h.d.}(\rho^a, \rho^b, \rho^c)-\kappa.$

\section{Main Results}
\label{sec:MainRes}

The main results of the paper regarding the two setups introduced in
the earlier section are summarized in the following two theorems.
\begin{thm}\label{thm:mainres1}The generalized DMT of the
  $(a,b,c)$-relay channel with $a=b=p$ is given by
\begin{IEEEeqnarray*}{l}
d_{(p,p,c)}^*(r) =\\
\quad\quad\begin{cases}
(p-r)^++(c-r)^+ & \text{if } c\geq p,\\
p+c-2r & \text{if } c< p, r\leq c,\\
p-\frac{(p-c)r}{p-r} & \text{if } c<p, c < r < \frac{p}{2},\\
p+c-2r & \text{if } c<p, \frac{p}{2}\leq r\leq \frac{p+c}{2},
\end{cases}
\end{IEEEeqnarray*} 
where the optimal strategy for the first, second and fourth cases is
static QMF with a half-TX half-RX schedule. In the third case, the DMT
is achieved by a dynamic decode and forward strategy.
\end{thm}

The result of \cite{PawAveTse} corresponds to the special case
$(a,b,c)=(1,1,1)$, which falls in the first regime, while the result
of \cite{GunKhoGolPoo} corresponds to the special case
$(a,b,c)=(1,1,0)$, which falls in the third regime. The above results
uncover two other regimes where static QMF with a half-TX half-RX
schedule is optimal. Note that the performance reaches the full-duplex
DMT only in the first two regimes.

Theorem~\ref{thm:mainres1} suggests that the two strategies studied in literature, static QMF with a half-TX half-RX schedule and dynamic-decode-forward, are sufficient to achieve the optimal DMT in all regimes when the source to relay and relay to destination links have the same average SNR. A natural generalization of these two strategies
is dynamic QMF where a relay listens for a fraction of time determined
by its receive CSI that is not necessarily long enough to allow
decoding of the transmitted message. The relay then quantizes maps and
forwards the received signal as in the original QMF. In
Section~\ref{sec:anoteqb} we show that this additional flexibility for
the dynamic schedule can be critical for achieving the optimal DMT
when $a\neq b$. However, obtaining an explicit
expression for the optimal DMT or for the optimal dynamic schedule as
a function of the receive CSI at the relay seems difficult in this case. Instead, we demonstrate this numerically.

To obtain better insight on the necessity of dynamic QMF, we next turn
to the parallel relay network given in
Figure~\ref{subfig:parallel}. We show that even in this simple case
with no broadcast or superposition of signals, dynamic QMF is needed
to achieve the optimal trade-off. In this case, we explicitly
characterize the optimal trade-off and the optimal dynamic schedule at
the relays.%

\begin{thm}\label{thm:mainres2} The DMT of the parallel relay network in Figure~\ref{subfig:parallel} (in which the  average SNR's of the four orthogonal links are all equal to each other) is given by 
$$
d^*(r)= 	\begin{cases}
	2-\frac{r}{1-r},& 0\leq r< \frac{1}{2}\\
	2(1-r), & \frac{1}{2}\leq r\leq 1.
	\end{cases}
$$
where in the first case the optimal DMT is achieved by a dynamic QMF
scheme and in the second case it is achieved by a static QMF scheme
with a half-TX half-RX schedule. In both regimes DDF is sub-optimal.
\end{thm}

Section~\ref{sec:relay} and \ref{sec:parallel} are devoted to the
proofs of the two theorems.

\section{The half-duplex $(a,b,c)$-relay channel when $a= b$}\label{sec:relay}

In this section, we prove Theorem~\ref{thm:mainres1} in a number of
steps. Each step is summarized in a lemma.

\subsection{The Full-Duplex DMT}

We first derive the generalized diversity-multiplexing tradeoff of the
full-duplex $(a,b,c)$-relay channel. This serves as an upper bound for
the optimal DMT of the corresponding half-duplex channel.

\begin{lem} \label{lem:fullduplex} The diversity-multiplexing tradeoff
  of the full-duplex $(a,b,c)$-relay channel is given
  by $$d_{f.d.}(r)=\left(\min(a,b)-r\right)^{+}+(c-r)^{+}.$$ where
  $a^+=\max(a,0)$.
\end{lem}
\begin{proof}
For the full-duplex relay channel at hand, since reliable communication at rates larger than the upper bound $C_u(\rho^a,
\rho^b, \rho^c)$ given in \eqref{eq:fdcutset} is
fundamentally impossible, the error probability of any strategy is
lower bounded by some fixed $\epsilon>0$ when the target rate at the
transmitter $r\log \rho$ turns out to be larger than $C_u(\rho^a,
\rho^b, \rho^c)$. Therefore, the probability of error for any strategy
is lower bounded by
$$ P_{e}(\rho^a, \rho^b, \rho^c) \geq \,\epsilon\,\, \mathbb{P}(C_u(\rho^a, \rho^b, \rho^c)\leq r\log\rho),$$
when the channel realizations are not known at the transmitter. The
probability is calculated over the random channel
realizations. Therefore, the diversity multiplexing tradeoff of the
full-duplex relay channel can be upper bounded by $d_{f.d.}(r)\leq
d_u(r),$ where $$d_u(r) =
-\lim_{\rho\rightarrow\infty}\frac{\log\mathbb{P}(C_u(\rho^a, \rho^b,
  \rho^c)\leq r\log\rho)}{\log\rho}.$$

The achievable rate by QMF for the full-duplex relay channel is given in \eqref{eq:fdqmf} as $C_u - \kappa$.  Since we assume that the codeword lengths are sufficiently large (in the quasi-static
model), the probability of error for QMF can be upper bounded as
follows:
\begin{IEEEeqnarray*}{lCl}
  \text{Pr}(\text{error})\\ 
   =  \text{Pr}(C_u - \kappa \leq r\log\rho)\cdot\text{Pr}(\text{error}|C_u - \kappa \leq r\log\rho) \\
  \quad\quad + \;\text{Pr}(C_u - \kappa > r\log\rho)\cdot\text{Pr}(\text{error}|C_u - \kappa_1 > r\log\rho)\\
   \leq  \text{Pr}(C_u - \kappa \leq r\log\rho) + \text{Pr}(\text{error}|C_u - \kappa > r\log\rho)\\
   \leq  \text{Pr}(C_u - \kappa \leq r\log\rho)+\epsilon,
\end{IEEEeqnarray*}
$\forall\epsilon>0$, where the last inequality follows since $\text{Pr}(\text{error}|C_u -
\kappa > r\log\rho)$ can be made arbitrarily small by choosing a
sufficiently long codeword length.

Since QMF is one particular scheme, the diversity achieved by QMF
$d_{QMF}(r)$ is a lower bound on $d_{f.d.}(r).$ We can prove that QMF
achieves the optimal DMT as follows:
\begin{IEEEeqnarray*}{rCl}
  d_{f.d.}(r) & \geq & d_{QMF}(r)\\ 
  & \geq & -\lim_{\rho\rightarrow\infty}\frac{\log\mathbb{P}(C_u(\rho^a, \rho^b, \rho^c) - \kappa \leq r\log\rho)}{\log\rho}\\
  & = & -\lim_{\rho\rightarrow\infty}\frac{\log\mathbb{P}(C_u(\rho^a, \rho^b, \rho^c) \leq r\log\rho)}{\log\rho}\\
  & = & d_u(r)\\
  & \geq & d_{f.d.}(r),
\end{IEEEeqnarray*}
and hence $d_{QMF}(r) = d_u(r) = d_{f.d.}(r).$ 

This equality, apart from showing that QMF is optimal, is also
convenient from the point of view of characterizing the optimal DMT:
the equality of $d_{f.d.}(r)$ and $d_u(r)$ implies that we can define
outage as the event when the cutset bound (instead of the capacity)
falls below the transmission rate $r\log\rho.$ This is convenient
because we have an explicit expression for the cutset bound whereas
the capacity of the relay channel is not known.\footnote{These
  arguments also hold for general Gaussian relay networks.} Thus, we have the chain of equalities given at the top of the next page, in which $(a)$ follows because in the high
SNR limit, the event $$C_{u}(\rho^a,\rho^b,\rho^c)\leq r\log \rho$$
 is equivalent to
$$\min\left(\max(\alpha^+,\gamma^+),\max(\beta^+,\gamma^+)\right) \leq
r,$$ where $\alpha, \beta, \gamma$ are defined in \eqref{eq:defalphabeta}
and $(b)$ follows by plugging in the
expression for the joint pdf $p_{\alpha,\beta,\gamma}$ and simplifying
in a manner similar to \cite{ZheTse}.

  \begin{figure*}[!th]
  \normalsize
\begin{IEEEeqnarray*}{rCl}
  d_{f.d.}(r) & = & -\lim_{\rho\rightarrow\infty}\frac{\log\mathbb{P}(C_u(\rho^a, \rho^b, \rho^c)\leq r\log\rho)}{\log\rho}\\
  & \stackrel{(a)}{=} & -\lim_{\rho\rightarrow\infty}\frac{\log\mathbb{P}(\min\left(\max(\alpha^+,\gamma^+),\max(\beta^+,\gamma^+)\right) \leq r)}{\log\rho}\\
  & = & -\lim_{\rho\rightarrow\infty}\frac{1}{\log\rho}\log\left(\int_{\min\left(\max(\tilde{\alpha}^+,\tilde{\gamma}^+),\max(\tilde{\beta}^+,\tilde{\gamma}^+)\right) \leq r}p_{\alpha,\beta,\gamma}(\tilde{\alpha},\tilde{\beta},\tilde{\gamma})\; d\tilde{\alpha}\;d\tilde{\beta}\;d\tilde{\gamma}\right)\\
  & \stackrel{(b)}{=} & \; \min_{\substack{\alpha\leq a,\; \beta\leq b,\; \gamma\leq c,\\ \min(\max(\alpha^+,\gamma^+),\max(\beta^+,\gamma^+))\leq r }} a+b+c-\alpha-\beta-\gamma\\
  & = & \;\min_{\substack{0\leq\alpha\leq a,\; 0\leq\beta\leq b,\;
      0\leq\gamma\leq c,\\
      \min(\max(\alpha,\gamma),\max(\beta,\gamma))\leq r }}\;
  a+b+c-\alpha -\beta -\gamma
\end{IEEEeqnarray*}
\hrulefill
\end{figure*}

So $d_{f.d.}(r)$ is given by the solution to the following optimization problem:
\begin{IEEEeqnarray}{l}
\min\;\;  a+b+c-\alpha-\beta-\gamma \nonumber\\
    \textrm{s.t. }\quad\min\left(\max(\alpha,\gamma),\max(\beta,\gamma)\right) \leq r,\nonumber\\
  \quad\quad 0\leq\alpha\leq a,\;0\leq \beta\leq b,\;0\leq \gamma\leq c.\label{eq:optprob}
\end{IEEEeqnarray}
We solve this optimization problem in the remainder of this proof. For the sake of brevity, define 
\begin{equation}
s(\alpha,\beta,\gamma)\triangleq a+b+c-\alpha-\beta-\gamma.\label{eq:s}
\end{equation}
\begin{itemize}
\item If $\gamma > \min(\alpha,\beta)$, then $${\min\left\{\max(\alpha,\gamma),\max(\beta,\gamma)\right\} =\gamma,}$$and hence the feasible region becomes ${\min(\alpha,\beta) < \gamma \leq r}$. It can be easily verified that the optimal solution is given by
$$\begin{array}{l}\gamma  =  \min(c,r), \\ \min(\alpha,\beta) =  \min(\gamma,a,b),\\ \max(\alpha,\beta) = \max(a,b)\end{array}$$
and \begin{IEEEeqnarray}{l}
s(\alpha,\beta,\gamma)\nonumber\\
\enskip = a+b+c-\max(a,b)-\min(a,b,c,r)-\min(c,r)\nonumber\\
\enskip =  \min(a,b)-\min(a,b,c,r)+(c-r)^+.\label{eq:FDcase1}\end{IEEEeqnarray}
\item If $\gamma \leq \min(\alpha,\beta)$, then $${\min\left\{\max(\alpha,\gamma),\max(\beta,\gamma)\right\} = \min(\alpha,\beta),}$$and outage implies ${\gamma\leq\min(\alpha,\beta)\leq r}$. The optimal solution in this case is $$\begin{array}{l}\min(\alpha,\beta)=\min(a,b,r),\\ \max(\alpha,\beta)=\max(a,b),\\ \gamma = \min(\alpha,\beta,c,r)\end{array}$$ and $s(\alpha,\beta,\gamma)$ has the value \begin{equation}\label{eq:FDcase2}(\min(a,b)-r)^+ + c - \min(a,b,c,r).\end{equation} \end{itemize} 
The optimal value of $s(\alpha,\beta,\gamma)$ is given by the minimum of \eqref{eq:FDcase1} and \eqref{eq:FDcase2} which is
$$d(r)=\left(\min(a,b)-r\right)^{+}+(c-r)^{+}.$$ \end{proof}

\subsection{QMF with a fixed schedule for the half-duplex relay}
\label{subsec:StaticQMF}
We now investigate the performance of the quantize map and forward
strategy (QMF) in \cite{AveDigTse} for the half-duplex relay: here,
the relay listens for half of the total duration for communication,
then quantizes its received signal at the noise level and maps it to a
random codeword, and transmits it in the second half. Since the TX-RX
schedule is fixed ahead of time and is independent of the
instantaneous channel realizations, we call this a \emph{static} QMF
strategy. Note that the strategy uses only receive CSI at the relay to
determine the noise level for quantization.

\begin{lem}\label{lem:statqmf}
The DMT achieved by \emph{static} QMF on the half-duplex \((a,b,c)\)-relay channel is given by
\begin{IEEEeqnarray*}{l}
d_{QMF}(r)= \\
\quad\quad\begin{cases}
	\left(\min(a,b)-r\right)^{+}+(c-r)^{+} & \text{if } c \geq \min(a,b) \\ 
	\left(\min(a,b)+c-2r\right)^{+} & \text{if } c < \min(a,b).
	\end{cases}
\end{IEEEeqnarray*}
\end{lem}

\begin{proof}
The rate $R_{QMF}$ achieved by QMF on the half-duplex relay channel in Figure~\ref{subfig:abc_relay} with a {\em fixed} RX-TX schedule for the relay is lower bounded by $R_{QMF}$ in \eqref{eq:hdqmf}. 
Hence, using the same line of arguments as in Lemma~\ref{lem:fullduplex}, we can reduce the problem of characterizing the DMT achieved by this strategy to the following optimization problem:
\begin{IEEEeqnarray}{C}
d_{QMF}(r) \;\;=\;\; \min_{(\alpha,\beta,\gamma) \in\mathcal{O}(r)}\;  s(\alpha,\beta,\gamma)\label{eq:statqmfoptprob}
\end{IEEEeqnarray} where $s(\alpha,\beta,\gamma)$ is defined in \eqref{eq:s} and

\begin{IEEEeqnarray*}{l}
\mathcal{O}(r) \\
= \left\{ (\alpha,\beta,\gamma) : 0\leq\alpha\leq a,\;0\leq \beta\leq b,\;0\leq \gamma\leq c,\; r_{h.d.} \leq r\right\}
\end{IEEEeqnarray*} and 
\begin{IEEEeqnarray}{l}
r_{h.d.}\triangleq\nonumber\\
\min\left\{t\max(\alpha,\gamma) +(1-t)\gamma,t\gamma+(1-t)\max(\beta,\gamma)\right\}.\IEEEeqnarraynumspace\label{eq:rhd}
\end{IEEEeqnarray}
The set $\mathcal{O}(r)$, as before, is the set of channel realizations for which the strategy is in outage, i.e. the  multiplexing rate $r_{h.d.}$ achieved by the strategy falls below  the desired multiplexing rate $r$. We choose $t=1/2$ for the strategy in which case the multiplexing rate $r_{h.d}$ becomes
\begin{equation*}\label{eq:rhd_halfhalf}
\min\left\{\frac{1}{2}\max(\alpha,\gamma) +\frac{1}{2}\gamma,\frac{1}{2}\gamma+\frac{1}{2}\max(\beta,\gamma)\right\}.\end{equation*}
We solve the optimization problem by splitting it into cases $c \geq\min(a,b)$ and $c <\min(a,b)$.

\begin{description}
\item[Case I ] \hspace{5 pt}$c \geq\min(a,b)$
\begin{itemize}[leftmargin=*]
\item If $\gamma > \min(\alpha,\beta)$, then ${r_{h.d.}\leq r}$ implies $$\min(\alpha,\beta) < \gamma \leq r.$$As in the proof for the full-duplex case, the optimal solution is $$\begin{array}{l}\gamma  =  \min(c,r), \\ \min(\alpha,\beta) =  \min(\gamma,a,b),\\ \max(\alpha,\beta) = \max(a,b),\end{array}$$ which gives \begin{IEEEeqnarray}{l}
s(\alpha,\beta,\gamma)\nonumber\\
=a+b+c-\max(a,b)-\min(a,b,c,r)-\min(c,r)\nonumber\\
=  \min(a,b)-\min(a,b,c,r)+(c-r)^+\nonumber\\
=  \left(\min(a,b)-r\right)^{+}+(c-r)^{+}\label{eq:HDcase11}\end{IEEEeqnarray}
\item If $\gamma \leq \min(\alpha,\beta)$, then $r_{h.d.}\leq r$ implies $$\frac{1}{2}\left(\gamma+\min(\alpha,\beta)\right)\leq r.$$
\begin{itemize}
\item If $r\leq\min(a,b)$, then an optimal point is
$$\begin{array}{l}\gamma=\min(\alpha,\beta)=r,\\ \max(\alpha,\beta)=\max(a,b),\end{array}$$  which gives $$s(\alpha,\beta,\gamma)=\min(a,b)+c-2r.$$
\item If $r>\min(a,b)$, then the optimal point is
$$\begin{array}{l}\gamma= \min(\alpha,\beta)=\min(a,b),\\ \max(\alpha,\beta)=\max(a,b),\end{array}$$ and this results in $$s(\alpha,\beta,\gamma)=c-\min(a,b).$$
\end{itemize} 
\end{itemize}
Combining these results it is easy to observe that \eqref{eq:HDcase11} is the optimal solution.
\item[Case II ] \hspace{7 pt}$c<\min(a,b)$
\begin{itemize}[leftmargin=*]
\item If $\gamma > \min(\alpha,\beta)$, then $r_{h.d.}\leq r$ implies $$\min(\alpha,\beta) < \gamma \leq r.$$
\begin{itemize}
\item If $r\leq c$, an optimal point is $$\begin{array}{l}\gamma=\min(\alpha,\beta)=r,\\ \max(\alpha,\beta)=\max(a,b),\end{array}$$ which gives $$s(\alpha,\beta,\gamma)=\min(a,b)+c-2r.$$
\item If $r>c$, then the optimal point is $$\begin{array}{l}\gamma= \min(\alpha,\beta)=c,\\ \max(\alpha,\beta)=\max(a,b),\end{array}$$ and this results in $$s(\alpha,\beta,\gamma)=\min(a,b)-c.$$
\end{itemize} 
\item If $\gamma \leq \min(\alpha,\beta)$, then $r_{h.d.}\leq r$ implies $$\frac{1}{2}\left(\gamma+\min(\alpha,\beta)\right)\leq r.$$
\begin{itemize}
\item If $r\leq c$, an optimal point is $$\begin{array}{l}\gamma=\min(\alpha,\beta)=r,\\ \max(\alpha,\beta)=\max(a,b),\end{array}$$ which gives $$s(\alpha,\beta,\gamma)=\min(a,b)+c-2r.$$
\item If $ r > c$, an optimal point is $$\begin{array}{l}\gamma = c,\\ \min(\alpha,\beta)=2r-c,\\ \max(\alpha,\beta)=\max(a,b),\end{array}$$ and this results in $$s(\alpha,\beta,\gamma)=\min(a,b)+c-2r.$$
\end{itemize}
\end{itemize}
Thus, $d(r) = \left(\min(a,b)+c-2r\right)^+$ in this case.
\end{description}\end{proof}

Comparing Lemma~\ref{lem:fullduplex} and Lemma~\ref{lem:statqmf}, we immediately have the following corollary.
\begin{cor}\label{cor:sqmf}
Static QMF is optimal and achieves the full duplex DMT in the half-duplex $(a,b,c)$ relay channel when
\begin{itemize} 
\item $c\geq\min(a,b)$,
\item $c<\min(a,b)$ and $r\leq c$.
\end{itemize}
\end{cor}

This result shows that the half-duplex constraint does not appear in
the optimal DMT as long as $c\geq\min(a,b)$, \emph{i.e.,} when the
average SNR of the direct link is larger than the average SNR of one
of the relay links. The full-duplex DMT can be achieved with a fixed
schedule, extending the result of \cite{PawAveTse} for
$(a,b,c)=(1,1,1)$. The same conclusion holds for $c<\min(a,b)$ but
only for small multiplexing rates, i.e. when $r\leq c$.

\subsection{Dynamic Decode and Forward}
Since $c\geq\min (a,b)$ has been completely characterized, we focus on
the case $c<\min(a,b)$ in the rest of this section. We next establish
the DMT achieved by dynamic decode and forward (DDF) introduced in
\cite{AzaGamSch}. Here the relay node waits until it it is able to
decode the transmitted message from the source which is encoded with a
random Gaussian codebook. It then re-encodes the message via a
randomly chosen independent Gaussian codebook and transmits it in the
remaining time. The destination node chooses the most likely message
in the source codebook given its observation. The fraction of time the
relay listens is determined dynamically depending on the transmission
rate and the instantaneous realization of the S-R link.

Let $\alpha,\beta,\gamma$ be as defined in \eqref{eq:defalphabeta}. Following \cite{AzaGamSch}, the fraction of time the relay needs to listen to decode the source message is given by $t = \frac{r\log\rho}{\log(1+|h_{sr}|^2\rho^a)} \rightarrow \frac{r}{\alpha}$ asymptotically in $\rho$. Outage occurs if at least one of the following two events occur:
\begin{itemize}
\item $\frac{r}{\alpha}> 1$ and $\gamma < r$: In this case, the relay never gets to decode the source message and therefore never gets the chance to transmit, and the direct link is not strong enough to support the desired rate alone; 
\item $t = \frac{r}{\alpha}\leq 1$ and $t\gamma+(1-t)\max(\gamma,\beta) < r$: In this case, the relay decodes and transmits but the mutual information acquired over the S,R--D cut is not sufficient to support the desired rate.
\end{itemize}
As before, the DMT of this strategy is given by $d_{DDF}(r) = \min\; a+b+c-\alpha-\beta-\gamma$ given the system is in outage. 
Solving this optimization problem, we arrive at the following lemma. 

\begin{lem}\label{lem:ddf}
The diversity-multiplexing tradeoff achieved by DDF on the half-duplex \((a,b,c)\)-relay channel when ${c<\min(a,b)}$ is given by $\max(d_{DDF}(r),0)$ where $d_{DDF}(r)=$
$$
\begin{cases}
	\min(a,b)+c-2r & \text{if } 0 \leq r \leq \min\left(c,\frac{\max(a,b)}{2}\right), \\
	\vspace{2mm}
	\min(a,b) - \frac{(\max(a,b)-c)r}{\max(a,b)-r} & \text{if } c < r < \frac{\max(a,b)}{2},\\ \vspace{2mm}
	\left(\frac{ab}{r}-a-b+c\right)^+ &\text{if } r\geq \frac{\max\left(a,b\right)}{2}.
	\end{cases}
$$
\end{lem}
Note that in the first regime DDF also achieves the full-duplex DMT. Note also that the second regime only occurs when $c\leq \frac{\max(a,b)}{2}$. 
\begin{proof}
Please refer to Appendix \ref{sec:lemddf}.
\end{proof}

\subsection{DMT with global CSI} 
We next turn to proving upper bounds on the achievable DMT that are tighter than the full duplex upper bound. In this section, we upper bound the achievable DMT under the optimistic assumption that the relay not only knows its incoming channel state but all the channel states in the network and can optimize its TX-RX times accordingly (global CSI). This obviously upper bounds the achievable DMT when the relay only has receive CSI which is the assumption in our model. (In the next section, we derive an even tighter upper bound on the achievable DMT with only receive CSI at the relay.)    

In the current and the next subsections, we restrict our attention to the case when $a$ and $b$ are equal. Let $a=b=p$.\footnote{Extending our results to the case $a\neq b$ remains an open problem; see Section \ref{sec:anoteqb} for a discussion.}  Recall that we are considering $c<p$ since when $c\geq p$ we have shown in the earlier sections that static QMF is DMT optimal. The upper bound of the current section, establishes yet another regime where static QMF with half TX-half RX schedule achieves the optimal DMT. We show that when $r\geq \frac{p}{2}$, static QMF achieves the optimal DMT although it falls short of achieving the full-duplex performance.

\begin{lem}\label{lem:dqmf}
The DMT of the half-duplex $(a,b,c)$-relay channel with global CSI $d_{G-CSI}(r)$ is given by
\begin{equation*}\label{eq:dqmf}
d_{G-CSI}(r)= \min_{(\alpha,\beta,\gamma)\in\mathcal{O}(r)} s(\alpha,\beta,\gamma),
\end{equation*}
\begin{equation}\label{eq:globaloutage}
\mathcal{O}(r) = \left\{(\alpha,\beta,\gamma): \begin{array}{cc}\frac{\alpha\beta-\gamma^2}{\alpha+\beta-2\gamma} \leq r,\\ \gamma<\min(\alpha,\beta),\end{array} \begin{array}{cc}0\leq\alpha\leq a, \\ 0\leq\beta\leq b, \\ 0\leq\gamma\leq c\end{array}\right\}.\end{equation}
\end{lem}
We do not explicitly characterize the trade-off here. In Lemma~\ref{lem:abstatqmf} below, we will further upper bound $d_{G-CSI}(r)$ by considering one specific point in the domain of this minimization problem.
\begin{proof} From Appendix~\ref{app:glob}, the DMT under global CSI is given by the solution of the optimization problem:
\begin{IEEEeqnarray}{rCl}
d_{G-CSI}(r) & = & \min_{(\alpha,\beta,\gamma)\in\mathcal{O}(r)}\; a+b+c-\alpha-\beta-\gamma,\IEEEeqnarraynumspace\label{eq:globalopt}
\end{IEEEeqnarray}
where
\begin{IEEEeqnarray*}{l}
\mathcal{O}(r)=\left\{(\alpha,\beta,\gamma) : \begin{array}{l}0\leq\alpha\leq a,\;0\leq\beta\leq b,\;0\leq\gamma\leq c,\\ \max_{t(\alpha,\beta,\gamma)} r_{h.d.}\leq r \end{array}\right\},
\end{IEEEeqnarray*}
where $r_{h.d.}$ is defined in \eqref{eq:rhd} but now we allow $t$ to depend on $\alpha, \beta$ and $\gamma$.

If we take $\gamma\geq\min(\alpha,\beta)$ in $\mathcal{O}(r)$, the right-hand side of \eqref{eq:globalopt} is greater than or equal to $ d_{f.d.}(r)$ and the bound is no tighter than the full-duplex upper bound. So, we concentrate on $\gamma<\min(\alpha,\beta)$. It is easy to see that the optimal choice of $t$ is obtained by equating the two terms in $r_{h.d.}$ and  when $\gamma<\min(\alpha,\beta)$, the optimal listening time for the relay becomes $$t=\frac{\beta-\gamma}{\alpha+\beta-2\gamma}.$$
Substituting this in $r_{h.d.}$ gives the outage region in \eqref{eq:globaloutage}. This completes the proof of the lemma.\end{proof}

\begin{lem}\label{lem:abstatqmf}
When $c<a=b=p$, static QMF (with equal listening and transmit times) is optimal for $r\geq \frac{p}{2}$ on the half-duplex $(a,b,c)$-relay channel.
\end{lem}
\begin{proof} A critical outage event for the static QMF protocol for $r\geq\frac{p}{2}$ when $a=b=p$ is $(\alpha,\beta,\gamma)=\left(p,p,2r-p\right)$. It can be verified that $(p,p,2r-p)\in \mathcal{O}(r)$ in \eqref{eq:globaloutage}. Therefore, $d_{G-CSI}(r)\leq p+c-2r$ which is achieved by static QMF.\end{proof}

\subsection{DMT with local receive CSI}
We next establish an upper bound on the optimal DMT when the relay has only receive CSI. This upper bound shows that DDF is optimal under receive CSI in the range $c<r<\frac{p}{2}.$ To the best of our knowledge, this is the first upper bound on the optimal DMT under limited CSI.

\begin{lem}\label{lem:ddfreceivecsi}
When $a=b=p$ and $c < r < \frac{p}{2}$, the optimal DMT of the half-duplex $(a,b,c)$-relay channel with receive CSI is attained by DDF.
\end{lem}
\emph{Proof of Lemma~\ref{lem:ddfreceivecsi}:} 
From Appendix~\ref{app:local}, the DMT under local receive CSI $d_{L-CSI}(r)$ is given by the following: 
\begin{equation}\label{eq:localopt}
d_{L-CSI}(r) = \min_{\alpha\in[0,p]}\;\max_{t\in[0,1]}\;\min_{(\beta,\gamma)\in\mathcal{O}(r,\alpha,t)} s(\alpha,\beta,\gamma)
\end{equation} where 
$$\mathcal{O}(r,\alpha,t)=\left\{(\beta,\gamma): 0\leq\beta\leq p, \; 0\leq\gamma\leq c,\; r_{h.d.}\leq r\right\}.$$
This can be interpreted as follows. Nature chooses some $\alpha$ which we can observe and optimize $t$ accordingly (due to receive CSI); however, nature gets a second round in which it can make adversarial choices for $(\beta,\gamma)$ depending on $\alpha$ and $t$. In other words, the RX time $t(\alpha)$ chosen by the relay should work equally well for all possible realizations of $(\beta,\gamma)$. This creates the following tension: if $t$ is chosen very small, so that the relay cannot decode the source message, the communication can be in outage if the S--D link turns out to be weak,  in which case we may not be able to convey sufficient mutual information over the S-$\{\text{R,D}\}$ cut; whereas if we choose $t$ to be very large, so that the relay is left with little time to transmit, the R-D link can take on values that make the $\{\text{S,R}\}$--D cut sufficiently weak so as to cause outage. This intuition is formalized in the following analysis.
We fix $\alpha = p$ to get:
$$
d_{L-CSI}(r) \leq \max_{t\in[0,1]}\;\min_{(\beta,\gamma)\in\mathcal{O}(r,p,t))} s(p,\beta,\gamma).
$$
\begin{itemize}
\item If $t < r/p$, $(\beta,\gamma)=(p,0)$ is a feasible point in the above minimization problem. Hence $d(r)\leq c.$
\item If $t\geq r/p$, $\beta=\min\left(\frac{r-tc}{1-t},p\right)$, $\gamma=c$ is feasible. Hence we get 
$$
d(r) \leq p -\min\left(\frac{r-tc}{1-t},p\right) \leq p - \frac{(p-c)r}{p-r},
$$
if $ c<r<\frac{p}{2}.$
\end{itemize}
This shows that $d_{L-CSI}(r)\leq \max\left(c, p - \frac{(p-c)r}{p-r}\right)$. Since DDF achieves $p - \frac{(p-c)r}{p-r}$ which is equal to $\max\left(c, p - \frac{(p-c)r}{p-r}\right)$ for $c<r<\frac{p}{2}$, this proves that DDF is optimal for $c<r<\frac{p}{2}$  when the relay has only receive CSI.~\hfill$\QED$

The various lemmas established in this section complete the proof of Theorem \ref{thm:mainres1}.

\section{The half-duplex $(a,b,c)$-relay channel when $a\neq b$} \label{sec:anoteqb}
 A natural generalization of the two strategies, DDF and static QMF discussed in the earlier section, is \emph{dynamic} QMF where the relay listens for a fraction of time
determined by its receive CSI that is not necessarily long enough to
allow decoding. It then quantizes maps and forwards the received
signal as in the static QMF. This generalization was not necessary in
the earlier section when $a=b=p$. In this section, we give numerical
evidence to show that dynamic QMF is needed to achieve the optimal
trade-off when $a\neq b$ and $c<\min(a,b)$.

In Section \ref{subsec:StaticQMF}, we saw that a particular static QMF
scheme where the relay follows a half Rx-half Tx schedule is optimal
in certain regimes. However, more generally, static QMF schemes can
have schedules dependent on the multiplexing gain and the values of $a,b,c$, since these are assumed to be known to all nodes apriori. Hence, we can obtain the DMT
achieved by the best static QMF scheme by allowing the listening time
$t$ to be chosen optimally in the optimization problem in
\eqref{eq:statqmfoptprob}, which gives us the following optimization
problem:
\begin{equation}\label{eq:optstatqmf}
d_{SQMF}(r) = \max_t\min_{(\alpha,\beta,\gamma)\in\mathcal{O}(r,t)} s(\alpha,\beta,\gamma),
\end{equation}
\begin{IEEEeqnarray*}{l}\mathcal{O}(r,t)=\\
\left\{(\alpha,\beta,\gamma):\begin{array}{l} 0\leq\alpha\leq a,\; 0\leq\beta\leq b, \; 0\leq\gamma\leq c,\\ r_{h.d.}\leq r\end{array}\right\},\end{IEEEeqnarray*}where $r_{h.d.}$ is defined in \eqref{eq:rhd}. The formal derivation of this optimization problem is given in Appendix~\ref{app:stat}.

As described before, the DMT can be potentially improved if we allow schemes that are more general than static QMF or DDF, i.e. dynamic QMF, in which the listening time of the relay is allowed to also depend on the incoming channel realization but is not necessarily long enough that the relay can decode the message. The relay quantizes the received signal, maps it to a random codebook and forwards as in the original QMF. As described in Appendix~\ref{app:local}, the best such dynamic QMF scheme achieves the upper bound \eqref{eq:localopt}, hence:
\begin{equation}\label{eq:optdynqmf}
d_{DQMF}(r) = \min_{0\leq \alpha\leq a}\max_t\min_{(\beta,\gamma)\in\mathcal{O}(r,\alpha,t)} s(\alpha,\beta,\gamma),
\end{equation}
$$\mathcal{O}(r,\alpha,t)=\left\{(\beta,\gamma) : 0\leq\beta\leq b,\;0\leq\gamma\leq c,\; r_{h.d.}\leq r \right\}.$$

Obtaining an analytic solution for the optimization problems \eqref{eq:optstatqmf} and \eqref{eq:optdynqmf} is difficult. Instead, we evaluate the objective function on a fine grid in the feasible region for each $r$ and choose the best value numerically. The results are displayed in Figure \ref{fig:dqmf}. Also included is the exact DMT achieved by Dynamic-Decode-Forward (DDF), derived in Lemma~\ref{lem:ddf}.

\begin{figure}
\centering
%
%
%
%

\begin{tikzpicture}

\begin{axis}[%
width=2.8in,
height=1.4in,
xlabel=$r$,
ylabel=$d(r)$,
scale only axis,
xmin=0.5,
xmax=1.5,
ymin=0,
ymax=1.4,
legend style={at={(0.403320312500001,0.46)},anchor=south west,draw=black,fill=white,align=left}
]

\addplot [
color=black,
solid,
line width=1.0pt
]
coordinates{
(0.5,1.00000000673832)(0.520408163265306,0.968757876071901)(
0.540816326530612,0.937515747053664)(
0.561224489795918,0.905102040817105)(
0.581632653061224,0.871867734453418)(
0.602040816326531,0.839834668063791)(
0.622448979591837,0.805337519647058)(
0.642857142857143,0.769944341401626)(
0.663265306122449,0.737079247316704)(
0.683673469387755,0.700322234190335)(
0.704081632653061,0.667024704654938)(
0.724489795918367,0.628843537453758)(
0.744897959183674,0.595102040861679)(
0.76530612244898,0.56136054427255)(
0.785714285714286,0.521235521302867)(
0.806122448979592,0.487038058554062)(
0.826530612244898,0.452840595810307)(
0.846938775510204,0.418643133074904)(
0.86734693877551,0.384445671236903)(
0.887755102040816,0.350248208963839)(
0.908163265306122,0.31605074701767)(
0.928571428571429,0.281853285815463)(
0.948979591836735,0.247655825416693)(
0.969387755102041,0.213458366156241)(
0.989795918367347,0.179260908548921)(
1.01020408163265,0.145063430777719)(
1.03061224489796,0.110865968008829)(
1.05102040816327,0.0766685052400931)(
1.07142857142857,0.0424710425014463)(
1.09183673469388,0.00827358083176843)(
1.11224489795918,0)(
1.13265306122449,0)(
1.1530612244898,0)(
1.1734693877551,0)(
1.19387755102041,0)(
1.21428571428571,0)(
1.23469387755102,0)(
1.25510204081633,0)(
1.27551020408163,0)(
1.29591836734694,0)(
1.31632653061224,0)(
1.33673469387755,0)(
1.35714285714286,0)(
1.37755102040816,0)(
1.39795918367347,0)(
1.41836734693878,0)(
1.43877551020408,0)(
1.45918367346939,0)(
1.47959183673469,0)(
1.5,0)
};
\addlegendentry{Best Static QMF};

\addplot [
color=black,
dashed,
line width=1.0pt
]
coordinates{
(0.5,1.00806452991063)(
0.520408163265306,0.979838711960911)(
0.540816326530612,0.967741935661654)(
0.561224489795918,0.943548387174663)(
0.581632653061224,0.919354838727277)(
0.602040816326531,0.895161290460098)(
0.622448979591837,0.870967742039411)(
0.642857142857143,0.846774194435578)(
0.663265306122449,0.822580645397219)(
0.683673469387755,0.786290323902335)(
0.704081632653061,0.762096782503287)(
0.724489795918367,0.737903228571752)(
0.744897959183674,0.701612908524061)(
0.76530612244898,0.677419356903536)(
0.785714285714286,0.641129037718007)(
0.806122448979592,0.60630499438042)(
0.826530612244898,0.551267282151003)(
0.846938775510204,0.507680491551461)(
0.86734693877551,0.462180200222473)(
0.887755102040816,0.41647465437788)(
0.908163265306122,0.37288786482335)(
0.928571428571429,0.32930107526882)(
0.948979591836735,0.285714285714291)(
0.969387755102041,0.24137931034483)(
0.989795918367347,0.206896551724145)(
1.01020408163265,0.172413793103636)(
1.03061224489796,0.131720430107613)(
1.05102040816327,0.086290324479358)(
1.07142857142857,0.040860228975371)(
1.09183673469388,0)(
1.11224489795918,0)(
1.13265306122449,0)(
1.1530612244898,0)(
1.1734693877551,0)(
1.19387755102041,0)(
1.21428571428571,0)(
1.23469387755102,0)(
1.25510204081633,0)(
1.27551020408163,0)(
1.29591836734694,0)(
1.31632653061224,0)(
1.33673469387755,0)(
1.35714285714286,0)(
1.37755102040816,0)(
1.39795918367347,0)(
1.41836734693878,0)(
1.43877551020408,0)(
1.45918367346939,0)(
1.47959183673469,0)(
1.5,0)
};
\addlegendentry{Best Dynamic QMF};

\addplot [
color=black,
dotted,
line width=1.0pt
]
coordinates{
(0.5,1)(
0.520408163265306,0.972413793103448)(
0.540816326530612,0.944055944055944)(
0.561224489795918,0.914893617021277)(
0.581632653061224,0.884892086330935)(
0.602040816326531,0.854014598540146)(
0.622448979591837,0.822222222222222)(
0.642857142857143,0.789473684210526)(
0.663265306122449,0.755725190839695)(
0.683673469387755,0.72093023255814)(
0.704081632653061,0.68503937007874)(
0.724489795918367,0.648)(
0.744897959183674,0.609756097560976)(
0.76530612244898,0.570247933884297)(
0.785714285714286,0.529411764705882)(
0.806122448979592,0.487179487179487)(
0.826530612244898,0.443478260869565)(
0.846938775510204,0.398230088495575)(
0.86734693877551,0.351351351351352)(
0.887755102040816,0.302752293577982)(
0.908163265306122,0.252336448598131)(
0.928571428571429,0.2)(
0.948979591836735,0.145631067961165)(
0.969387755102041,0.0891089108910887)(
0.989795918367347,0.0303030303030303)(
1.01020408163265,0)(
1.03061224489796,0)(
1.05102040816327,0)(
1.07142857142857,0)(
1.09183673469388,0)(
1.11224489795918,0)(
1.13265306122449,0)(
1.1530612244898,0)(
1.1734693877551,0)(
1.19387755102041,0)(
1.21428571428571,0)(
1.23469387755102,0)(
1.25510204081633,0)(
1.27551020408163,0)(
1.29591836734694,0)(
1.31632653061224,0)(
1.33673469387755,0)(
1.35714285714286,0)(
1.37755102040816,0)(
1.39795918367347,0)(
1.41836734693878,0)(
1.43877551020408,0)(
1.45918367346939,0)(
1.47959183673469,0)(
1.5,0)
};
\addlegendentry{DDF};

\end{axis}
\end{tikzpicture}%
\caption{$(a,b,c)=(1.5,2,0.5)$}
\label{fig:dqmf}
\end{figure}
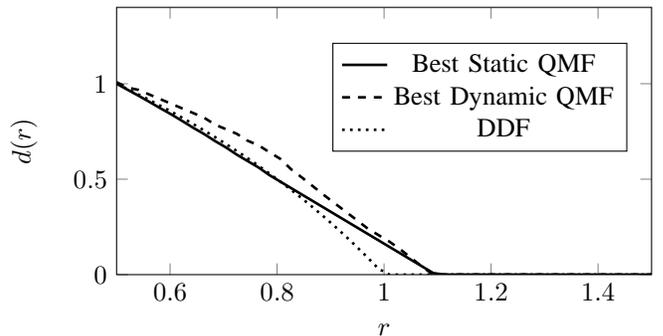

We can see from Figure \ref{fig:dqmf} that both static QMF and DDF fall short of the optimal tradeoff that the dynamic QMF scheme achieves. This suggests that it is not sufficient to consider static QMF or DDF to achieve the optimal DMT. While we are not able to provide an explicit expression for the optimal dynamic schedule or the optimal DMT in this case, identifying the optimal dynamic schedule is possible in some cases, as demonstrated in the next section.

\section{The Parallel Relay Network}\label{sec:parallel}
Dynamic QMF is applicable not only to the half-duplex relay channel but to more general half-duplex relay networks. In this section, we aim to demonstrate its necessity to achieve the optimal DMT of more general networks through the specific configuration of two parallel relays in Fig~\ref{subfig:parallel}. It is surprising that even though the parallel relay network is very simple since it does not involve broadcasting or superposition of signals, fixed schedules (with QMF) or simple decode-and-forward (with a dynamic schedule) are not sufficient  to achieve the optimal trade-off and a dynamic QMF strategy is needed. Recall from Section~\ref{sec:model} that we focus on the case when the exponential orders of the average received SNRs of all the four links are equal to $1$. This case is sufficient to demonstrate the necessity for dynamic QMF. Analogous to \eqref{eq:defalphabeta}, we define  
$$\alpha\triangleq \frac{\log(\left|h_{sr_1}\right|^2\rho)}{\log\rho},$$ $$\beta\triangleq \frac{\log(\left|h_{r_1d}\right|^2\rho)}{\log\rho},$$ $$\gamma\triangleq \frac{\log(\left|h_{sr_2}\right|^2\rho)}{\log\rho},$$ $$\delta\triangleq\frac{\log(\left|h_{r_2d}\right|^2\rho)}{\log\rho},$$
as the exponential orders of the instantaneous SNRs for the four links in the parallel relay network.

The difficulty in applying dynamic QMF and characterizing the trade-off it achieves is in identifying the optimal (dynamic) choice of the listening times at the relays. In the sequel, we identify an optimal choice for the listening time at the first relay as $t_1=1-\alpha(1-r)$. Similarly, $t_2=1-\gamma(1-r)$ for the second relay. Note that in the case of DDF, $t_1=r/\alpha$ which ensures that the relay can decode the transmitted message. The choice $t=1-\alpha(1-r)$, on the other hand, is motivated by the need to balance the multiplexing gain achieved over the two cuts of the network dynamically, based only on the observation of $\alpha$. Note that when $\alpha$ is large $t$ is small, and the strategy allocates more time to the second stage which helps in case the second stage turns out to be weak. When $\alpha$ is small, the relay allocates more time to listen. Indeed, if we were to apply this dynamic schedule $t=1-\alpha(1-r)$ to the $(1,1,c)$ half-duplex relay with $c\leq 1$, it can be readily observed that $1-\alpha(1-r)\geq r/\alpha$ when $\frac{r}{1-r}\leq\alpha\leq 1$ (which is the range of $\alpha$'s where DDF is not in outage) and so the relay can always decode the message. Moreover, in the critical events when $\alpha=\frac{r}{1-r}$ and $\alpha=1$ they allocate the same listening times for the relay. Indeed, it turns out that these two dynamic schedules are equivalent for the $(1,1,c)$ single relay channel. However, as we show in this section this is not the case for the parallel relay network. While dynamic QMF with  $t=1-\alpha(1-r)$ reaches the best achievable DMT with global CSI, DDF (and also static QMF) fails to do so.

In this section, we prove Theorem~\ref{thm:mainres2} in a number of steps summarized in lemmas. First, we establish an upper bound on the DMT of the half-duplex parallel relay network by allowing the switching times to depend on all channel realizations in a similar manner as Lemma~\ref{lem:dqmf}.

\begin{lem}\label{lem:dqmf2}
The DMT of the half-duplex parallel relay network with global CSI $d_{G-CSI}(r)$ is given by
\begin{equation}\label{eq:dqmf2}
d_{G-CSI}(r)= \min_{(\alpha,\beta,\gamma,\delta)\in\mathcal{O}(r)} 4-\alpha-\beta-\gamma-\delta
\end{equation} where $$\mathcal{O}(r) = \left\{(\alpha,\beta,\gamma,\delta): \begin{array}{lcc}\frac{\alpha\beta}{\alpha+\beta}+\frac{\gamma\delta}{\gamma+\delta} \leq r,\\0\leq\alpha,\beta,\gamma,\delta\leq 1 \end{array}\right\}.$$
\end{lem}
\begin{proof} 
The proof of this lemma follows on similar lines as the proof of Lemma~\ref{lem:dqmf}.

We first introduce the definition in \eqref{eq:cutset_parallel} on top of the next page, which when maximized over all $t_1(\alpha,\beta,\gamma,\delta)$ and $t_2(\alpha,\beta,\gamma,\delta)$ provides, within a constant gap, an upper bound on the capacity of the parallel relay network (see \cite[Section VI]{OzgDig2}).

\begin{figure*}[!th]
\normalsize
\begin{IEEEeqnarray}{rCl}
C^{parallel} &\triangleq & \min\left\{t_1\log(1+|h_{sr_1}|^2\rho) ,(1-t_1)\log(1+|h_{r_1d}|^2\rho)\right\}\nonumber\\
&& \quad\quad +\min\left\{t_2\log(1+|h_{sr_2}|^2\rho),(1-t_2)\log(1+|h_{r_2,d}|^2\rho)\right\}\label{eq:cutset_parallel}
\end{IEEEeqnarray}
\hrulefill
\end{figure*}


Now, following similar steps as in Lemma \ref{lem:dqmf}, we get the following upper bound on the DMT:
\begin{equation*}
\min_{(\alpha,\beta,\gamma,\delta)\in\mathcal{O}(r)}\;\; 4-\alpha-\beta-\gamma-\delta
\end{equation*}
where $\mathcal{O}(r)$ is given in \eqref{eq:outparallel} on top of the next page. 
\begin{figure*}[!th]
\normalsize
\begin{equation}\label{eq:outparallel}\left\{(\alpha,\beta,\gamma,\delta):\begin{array}{ccc}\max_{\substack{t_1(\alpha,\beta,\gamma,\delta),\\t_2(\alpha,\beta,\gamma,\delta)}}\left(\min(t_1\alpha,(1-t_1)\beta)+\min(t_2\gamma,(1-t_2)\delta)\right)\leq r,\\ 0\leq\alpha,\beta,\gamma,\delta\leq 1\end{array}\right\}\end{equation}
\hrulefill 
\end{figure*}
This follows since the outage event ${C^{parallel}\leq r\log\rho}$ is equivalent at high SNR to 
$$\left( \min(t_1\alpha^+,(1-t_1)\beta^+) + \min(t_2\gamma^+,(1-t_2)\delta^+)\right)\leq r.$$
The optimal choice for $t_1$ and $t_2$ is identified by setting ${t_1\alpha=(1-t_1)\beta}$ and ${t_2\gamma=(1-t_2)\delta}$. This results in \eqref{eq:dqmf2}.
\end{proof}

We now consider the dynamic QMF protocol described at the beginning of this section. Note that the listening times in this protocol depend only on the receive CSI.

\begin{lem}\label{lem:dqmf3}
The DMT obtained by a dynamic QMF protocol where listening times are chosen as $t_1 = 1-\alpha(1-r)$ and $t_2 = 1-\gamma(1-r)$ for $0\leq r\leq \frac{1}{2}$ is $d_{DQMF}(r)=2-\frac{r}{1-r}.$
\end{lem}
\begin{proof}
If the listening times are chosen as $t_1 = 1-\alpha(1-r)$ and $t_2 = 1-\gamma(1-r)$, the rate achievable by QMF can be lower-bounded \cite{AveDigTse} by 
$$C^{parallel}\Big|_{\substack{t_1 = 1-\alpha(1-r),\\ t_2 = 1-\gamma(1-r)}} -\kappa,$$ where $\kappa$ is a constant independent of the SNR and the channel realizations. This means that the multiplexing rate achieved by this protocol is 
\begin{IEEEeqnarray*}{l}\alpha\min\left(1-(1-r)\alpha,(1-r)\beta\right)\\
\quad\quad\quad\quad\quad\quad\quad\quad +\,\gamma\min\left(1-(1-r)\gamma,(1-r)\delta\right).
\end{IEEEeqnarray*} Hence we obtain that the DMT is given by the solution to the following optimization problem:
$$\min_{(\alpha,\beta,\gamma,\delta)\in\mathcal{O}(r)} 4-\alpha-\beta-\gamma-\delta$$
where $\mathcal{O}(r)$ for this lemma is given in \eqref{eq:outparalleldqmf} on top of the next page.
\begin{figure*}[!th]
\normalsize
\begin{equation}\label{eq:outparalleldqmf}
\left\{(\alpha,\beta,\gamma,\delta): \begin{array}{ccc}\alpha\min\left(\frac{1}{1-r}-\alpha,\beta\right)+\gamma\min\left(\frac{1}{1-r}-\gamma,\delta\right) \leq \frac{r}{1-r},\\0\leq\alpha,\beta,\gamma,\delta\leq 1 \end{array}\right\}
\end{equation}
\hrulefill
\end{figure*}
We solve this problem by considering three different cases:

\begin{itemize}
\item Case (i) \hspace{5 pt} $\alpha+\beta \leq \frac{1}{1-r}, \gamma+\delta\leq\frac{1}{1-r}$\\
The outage region becomes
$$
\mathcal{O}(r) = \left\{(\alpha,\beta,\gamma,\delta): \begin{array}{ccc}\alpha\beta+\gamma\delta \leq \frac{r}{1-r},\\ \alpha+\beta \leq \frac{1}{1-r},\\ \gamma+\delta\leq\frac{1}{1-r}, \\ 0\leq\alpha,\beta,\gamma,\delta\leq 1 \end{array}\right\}
$$
For any value of $(\gamma,\delta)$ such that $\gamma\delta<\frac{r}{1-r}$, the feasible region in the $(\alpha,\beta)$ plane looks as shown in Figure~\ref{subfig:par_casei}. It can be seen that to minimize the objective function, either $\alpha = 1$ or $\beta = 1$. Without loss of generality, assume $\alpha=1$. Similarly, we can conclude that for any value of $(\alpha,\beta)$, either $\gamma=1$ or $\delta=1$. Assume $\gamma = 1$. Thus, the optimization problem reduces to
$$\min_{{O}(r)}\;2-\beta-\delta
$$ where the outage region is
$$
\mathcal{O}(r) = \left\{(\beta,\delta): \begin{array}{ccc}\beta+\delta \leq \frac{r}{1-r},\\0\leq\beta,\delta\leq 1 \end{array}\right\}
$$
Hence, $d(r)=2-\frac{r}{1-r}.$
\begin{figure}
\centering
\subfigure[]{
\begin{tikzpicture}[scale=4]
\clip (-0.25,-0.25) rectangle (1.7,1.7);
\draw [fill,green!20] [domain=0.5:1] plot (\x, {0.5/\x}) -- (1,0) -- (0,0) -- (0,1) -- cycle;
\draw [white,pattern=north east lines] (0,0) -- (0,1) -- (0.7,1) -- (1,0.7) -- (1,0) -- cycle;
\path[draw,->,>=stealth',name path=x axis] (-0.25,0) -- (1.5,0);
\path[draw,->,>=stealth',name path=y axis] (0,-0.25) -- (0,1.5);
\coordinate (a) at (1,0);
\coordinate (b) at (0,1);
\coordinate[label=below left:{$\alpha$}] (temp1) at (1.5,0);
\coordinate[label=below left:{$\beta$}] (temp2) at (0,1.5);
\draw let \p1=(a) in (\x1,1pt) -- (\x1,-1pt) node[anchor=north,fill=white] {$1$};
\draw let \p1=(b) in (1pt,\y1) -- (-1pt, \y1) node[anchor=east,fill=white] {$1$};
\draw[dashed] (a) -- (1,1);
\draw[dashed] (b) -- (1,1);
\draw (0.4,1.3) -- (1.3,0.4);
\draw[domain=0.4:1.25] plot (\x, {0.5/\x});
\node[] (temp) at (1.35,0.8) {$\alpha\beta\leq \frac{r}{1-r}-\gamma\delta$};
\draw[->,>=stealth',shorten >=3] (temp) -- (1.1,0.45);
\node[] (temp2) at (1,1.3) {$\alpha+\beta\leq \frac{1}{1-r}$};
\draw[->,>=stealth',shorten >=3] (temp2) -- (0.6,1.1);
\end{tikzpicture}
\label{subfig:par_casei}}
\subfigure[]{
\begin{tikzpicture}[scale=4]
\clip (-0.25,-0.25) rectangle (1.7,1.7);
\draw [fill,green!20] (0,0) -- (1,0) -- (1,1) -- (0,1) -- cycle;
\draw [fill,white] (0.7,0.7) circle [radius=0.76];
\draw (0.7,0.7) circle [radius=0.76];
\path[draw,->,>=stealth',name path=x axis] (-0.25,0) -- (1.5,0);
\path[draw,->,>=stealth',name path=y axis] (0,-0.25) -- (0,1.5);
\coordinate (a) at (1,0);
\coordinate (c) at (0,1);
\coordinate[label=below left:{$\alpha$}] (temp1) at (1.5,0);
\coordinate[label=below left:{$\gamma$}] (temp2) at (0,1.5);
\draw let \p1=(a) in (\x1,1pt) -- (\x1,-1pt) node[anchor=north,fill=white] {$1$};
\draw let \p1=(c) in (1pt,\y1) -- (-1pt, \y1) node[anchor=east,fill=white] {$1$};
\draw[dashed] (a) -- (1,1);
\draw[dashed] (c) -- (1,1);
\draw (0.4,0.4) -- (1,0.4);
\draw (0.4,0.4) -- (0.4,1);
\draw [pattern=north east lines] (0.4,0.4) -- (1,0.4) -- (1,1) -- (0.4,1) -- cycle;
\node at (0.7,0.7) [circle,draw,fill,inner sep=0pt,minimum size=1mm] {};
\node[] (temp) at (0.7,0.2) {$\left(\frac{1}{2(1-r)},\frac{1}{2(1-r)}\right)$};
\draw[->,>=stealth',shorten >=3] (temp) -- (0.7,0.69);
\node[] (temp2) at (0.4,1.1) {$\left(\frac{r}{1-r},1\right)$};
\end{tikzpicture}
\label{subfig:par_caseiii}}
\caption{Proof of Theorem~\ref{lem:dqmf3}; feasible region is the intersection of shaded and dashed regions (a) Case i, (b) Case iii: no feasible point}
\end{figure}
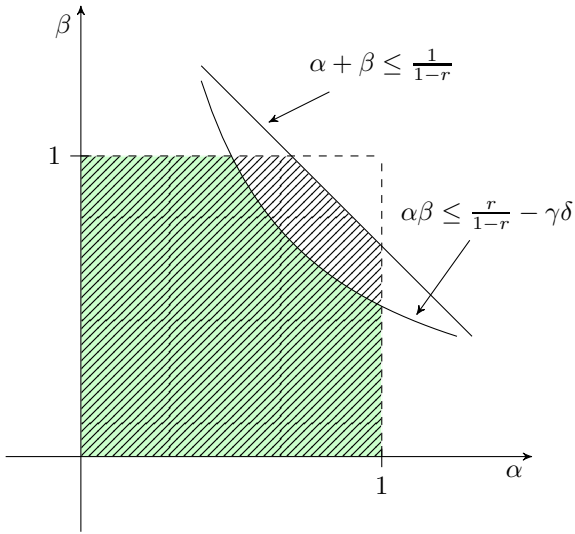
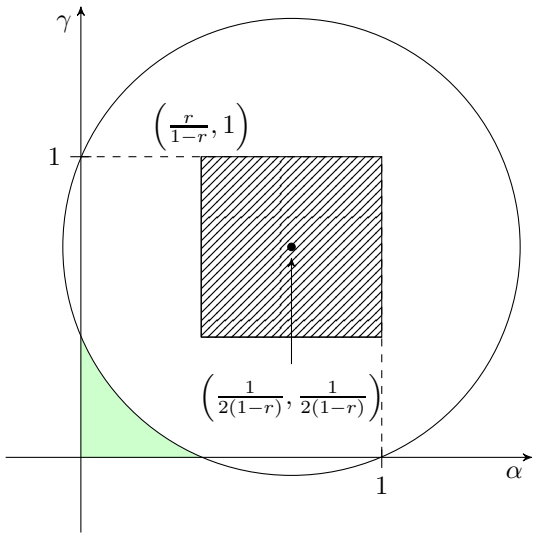

\item Case (ii) \hspace{5 pt} $\alpha+\beta \leq \frac{1}{1-r}, \gamma+\delta > \frac{1}{1-r}$\\
The outage region for this case is 
$$
\mathcal{O}(r) = \left\{(\alpha,\beta,\gamma,\delta): \begin{array}{ccc}\alpha\beta+\gamma\left(\frac{1}{1-r}-\gamma\right) \leq \frac{r}{1-r},\\ \alpha+\beta \leq \frac{1}{1-r},\\ \gamma+\delta>\frac{1}{1-r}, \\ 0\leq\alpha,\beta,\gamma,\delta\leq 1 \end{array}\right\}
$$
Firstly, we can immediately see that $\delta = 1$. Also, \begin{equation}\label{eq:gammabigger}\gamma +\delta>\frac{1}{1-r},\;\delta\leq 1\implies \gamma > \frac{r}{1-r}.\end{equation} Now we examine the first condition in the definition of the outage region:
$$\alpha\beta+\gamma\left(\frac{1}{1-r}-\gamma\right) \leq \frac{r}{1-r}\implies \alpha\beta\leq \gamma^2-\frac{\gamma-r}{1-r}.$$
The term on the RHS of the inequality needs to be non-negative to ensure that there exist feasible $\alpha,\beta$. 
$$\gamma^2-\frac{\gamma-r}{1-r}\geq 0 \Leftrightarrow \gamma \leq \frac{r}{1-r} \text{ or } \gamma \geq 1.$$
Combining this with \eqref{eq:gammabigger} and $0\leq\gamma\leq 1$, we can see that the only admissible value of $\gamma = 1$, which makes $\gamma^2-\frac{\gamma-r}{1-r} = 0$, which implies $\alpha=\beta=0$. Hence $d(r)=2.$
\item Case (iii) \hspace{5 pt} $\alpha+\beta > \frac{1}{1-r}, \gamma+\delta > \frac{1}{1-r}$\\
The outage region is:
$$
\left\{(\alpha,\beta,\gamma,\delta): \begin{array}{ccc}(\alpha +\gamma)\left(\frac{1}{1-r}-\gamma\right) \leq \frac{r}{1-r},\\ \alpha+\beta > \frac{1}{1-r},\\ \gamma+\delta>\frac{1}{1-r}, \\ 0\leq\alpha,\beta,\gamma,\delta\leq 1 \end{array}\right\}.
$$
As in the previous case, we note that $\beta=1$, $\delta=1$ and 
\begin{equation*}\label{eq:alphagammabigger}
\alpha > \frac{r}{1-r},\;\;\; \gamma > \frac{r}{1-r}
\end{equation*}
The first condition can be rewritten as 
\begin{IEEEeqnarray*}{l}
\left(\alpha-\frac{1}{2(1-r)}\right)^2 + \left(\gamma-\frac{1}{2(1-r)}\right)^2\\
\quad\quad\quad\quad\quad\quad\quad\quad\quad >\, \frac{1}{2}+\frac{1}{2}\left(\frac{r}{1-r}\right)^2.
\end{IEEEeqnarray*}
So the DMT is given by the following simplified optimization problem:
$$\min_{\mathcal{O}(r)}\; 2-\alpha-\gamma
$$ where $\mathcal{O}(r)$ is given in \eqref{eq:outcaseiii} on top of the next page. As shown in Figure~\ref{subfig:par_caseiii}, the circle of infeasible values of $(\alpha,\gamma)$ defined by the first condition contains the square of feasible values defined by the other two conditions for all $0\leq r\leq 1/2$. So there are no feasible values of $(\alpha,\gamma)$ which means that the optimization problem for this case is infeasible. 
\end{itemize}
Analysis of the remaining case $\alpha+\beta > \frac{1}{1-r}, \gamma+\delta\leq\frac{1}{1-r}$ is similar to Case (ii) by symmetry. Taking the minimum over all the cases, we have:
$$
d_{DQMF}(r) = 2 - \frac{r}{1-r},\;\;\;\; 0\leq r \leq \frac{1}{2}.
$$
\end{proof}

\begin{figure*}[!ht]
\normalsize
\begin{equation}\label{eq:outcaseiii}
\left\{(\alpha,\beta,\gamma,\delta): \begin{array}{ccc}\left(\alpha-\frac{1}{2(1-r)}\right)^2 + \left(\gamma-\frac{1}{2(1-r)}\right)^2 > \frac{1}{2}+\frac{1}{2}\left(\frac{r}{1-r}\right)^2, \\ \frac{r}{1-r}<\alpha,\gamma\leq 1 \end{array}\right\}
\end{equation}
\hrulefill
\end{figure*}

We now prove Theorem~\ref{thm:mainres2} via the following two lemmas.

First, Lemma~\ref{lem:par1} proves that the dynamic QMF protocol analyzed in the previous lemma achieves the upper bound established in Lemma~\ref{lem:dqmf2} for $0\leq r< 1/2$. We also show that static QMF and DDF are strictly suboptimal in this range of multiplexing gains.

For the remaining range of multiplexing gains $1/2 \leq r\leq 1$, Lemma~\ref{lem:par2} at the end of this section shows that the DMT of any scheme that depends only on receive CSI is upper bounded by $2-2r$ which is achieved by static QMF, thus establishing the optimality of static QMF in the class of receive CSI schemes.

\begin{lem}\label{lem:par1}
The optimal DMT of the half-duplex parallel relay network with receive CSI in the range $0\leq r< 1/2$ is given by 
$$d(r) = 2-\frac{r}{1-r},$$ which is achieved by the dynamic QMF scheme described in the previous lemma.
\end{lem}
\begin{proof}
To prove this lemma, we only need to show that the critical outage point $(\alpha,\beta,\gamma,\delta)=\left(1,\frac{r}{1-r},1,0 \right)$ of the dynamic QMF protocol is a feasible point of \eqref{eq:dqmf2}. That is easy to check. Since $0\leq r< 1/2$, every component of $\left(1,\frac{r}{1-r},1,0 \right)$ is in $[0,1]$. Also, 
$$
\frac{\alpha\beta}{\alpha+\beta}+\frac{\gamma\delta}{\gamma+\delta} = \frac{1\cdot\frac{r}{1-r}}{1+\frac{r}{1-r}} + \frac{1\cdot0}{1 + 0} = r \leq r,
$$ and hence, $d_{G-CSI}(r)\leq 4 - 1 - \frac{r}{1-r} -1-0 = 2-\frac{r}{1-r},$ which is achieved by the dynamic QMF scheme described in Lemma~\ref{lem:dqmf3}. This establishes the optimality of the dynamic QMF protocol in Lemma \ref{lem:dqmf3} for $0\leq r< 1/2$.

We now argue that static QMF and DDF are both strictly suboptimal in this range of multiplexing gains.
\begin{itemize}
\item Static QMF\\
Consider choosing $t_1=t_2=1/2$ for the static QMF scheme, for which the DMT can be obtained by solving the following optimization problem. We will argue later that this is indeed an optimal static choice of the listening times.
\begin{equation*}
d_{SQMF}(r) = \min_{\mathcal{O}(r)} 4 - \alpha-\beta-\gamma-\delta
\end{equation*} where $\mathcal{O}(r) =$
$$ \left\{(\alpha,\beta,\gamma,\delta):\begin{array}{ccc}\frac{1}{2}\min(\alpha,\beta)+\frac{1}{2}\min(\gamma,\delta)\leq r,\\ 0\leq\alpha,\beta,\gamma,\delta\leq 1 \end{array}\right\}.
$$
Without loss of generality, we can assume that $\alpha = \gamma = 1$. That makes the outage condition: $\beta + \delta \leq 2r$. For all $r\in[0,1]$, we can satisfy this condition with equality for $\beta,\delta\in[0,1]$, which implies
\begin{equation}\label{eq:parstatqmf}
d_{SQMF}(r) = 2-2r.
\end{equation}
We now argue that $t_1=t_2=1/2$ is optimal.  If $t_1$ and $t_2$ are set to any other values, we get the DMT by solving the following optimization problem.
\begin{equation*}
\min_{\mathcal{O}(r)} 4 - \alpha-\beta-\gamma-\delta
\end{equation*} where $\mathcal{O}(r)$ is given in \eqref{eq:parallelstatqmf} on top of the next page.

\begin{figure*}[!ht]
\normalsize
\begin{equation}\label{eq:parallelstatqmf} \left\{(\alpha,\beta,\gamma,\delta):\begin{array}{ccc}\min(t_1\alpha,(1-t_1)\beta)+\min(t_2\gamma,(1-t_2)\delta)\leq r,\\ 0\leq\alpha,\beta,\gamma,\delta\leq 1 \end{array}\right\}
\end{equation}
\hrulefill
\end{figure*}

This gives us a worse DMT which can be seen as follows. Consider the case when $t_1$ and $t_2$ are both no more than $1/2$. In this case, a feasible point in the optimization problem is $$(\alpha,\beta,\gamma,\delta) = \left(\min\{\frac{r}{t_1+t_2},1\},1,\min\{\frac{r}{t_1+t_2},1\},1\right),$$ that results in the objective value being $$2 - 2\min\{\frac{r}{t_1+t_2},1\},$$which is no more than $2-2r$. The other choices for $t_1$ and $t_2$ can also be treated similarly to conclude that $t_1=t_2=1/2$ is indeed optimal.

\item DDF\\
In DDF, each relay waits until it can decode the transmitted message, i.e. $t_1=r/\alpha$ and $t_2=r/\gamma$. One of the outage events is $\alpha <r$ and $\gamma < r$, in which case none of the relays gets to transmit. Hence, $$d_{DDF}(r)\leq 2-2r.$$ An alternative strategy would be to split the information stream into two streams each of multiplexing gain $r/2$ and send them over the two orthogonal paths in the parallel relay network. Both the relays perform DDF on the corresponding stream i.e. $t_1=r/2\alpha$ and $t_2=r/2\gamma$. However, now communication is in outage if even one of the paths is in outage, e.g. $\alpha < r/2$ when the first relay does not get a chance to transmit. So, the diversity of this scheme at rate $r$ is no more than $4-\frac{r}{2}-1-1-1 = 1-r/2$, which is even worse than $2-2r$ for $0\leq r< 1/2$.
\end{itemize}
Thus, for $0\leq r< 1/2$, neither DDF nor static QMF is able to achieve the optimal DMT.
\end{proof}

\begin{figure}
\centering
\input{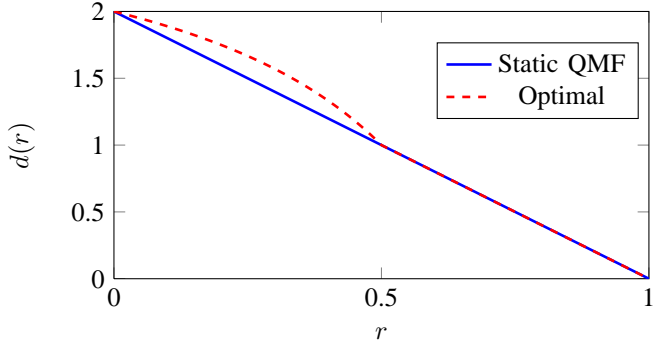}
\caption{DMT of the half-duplex parallel relay network}
\label{fig:dmt_parallel}
\end{figure}

\begin{lem}\label{lem:par2}
The optimal DMT of the half-duplex parallel relay network with receive CSI in the range $1/2\leq r\leq 1$ is given by
$$d(r)=2(1-r),$$ and is achieved by the static QMF scheme.
\end{lem}
\begin{proof}
As in Lemma \ref{lem:ddfreceivecsi}, the DMT of the half-duplex parallel relay network with receive CSI $d_{L-CSI}(r)$ is given by $d_{L-CSI}(r)=$
\begin{equation*}\label{eq:parallelreceiveCSI}
\min_{\alpha,\gamma}\;\max_{f(r,\alpha),f(r,\gamma)}\;\min_{(\beta,\delta)\in\mathcal{O}(r,\alpha,\gamma,f(r,\alpha),f(r,\gamma))} 4-\alpha-\beta-\gamma-\delta
\end{equation*} where $\mathcal{O}(r,\alpha,\gamma,t_1,t_2)=$ \begin{equation*}\left\{(\beta,\delta):\begin{array}{l}\min\left\{t_1\alpha,(1-t_1)\beta\right\}\\ \quad\quad +\,\min\left\{t_2\gamma,(1-t_2)\delta\right\}\leq r,\\ 0\leq\beta,\delta\leq 1\end{array}\right\},\end{equation*} and for explicitness $t_1$ and $t_2$ are set as $f(r,\alpha)$ and $f(r,\gamma)$, where $f(\cdot,\cdot)$ is any arbitrary function. Note that both relays use the same function $f(\cdot,\cdot)$ to decide the switching time by symmetry. Setting $\alpha=1$ and $\gamma=1$:
\begin{IEEEeqnarray}{l}
d_{L-CSI}(r)\nonumber\\  
\leq  \min_{\alpha,\gamma}\;\max_{f(r,\alpha),f(r,\gamma)}\;\min_{\mathcal{O}(r,\alpha,\gamma,f(r,\alpha),f(r,\gamma))} 4-\alpha-\beta-\gamma-\delta \nonumber \\
 \leq  \max_{f(r,1),f(r,1)}\;\min_{\mathcal{O}(r,1,1,f(r,1),f(r,1))} 2-\beta-\delta \nonumber
\end{IEEEeqnarray}
\begin{itemize}
\item If $f(r,1)\leq\frac{r}{2}$, then
$$(1,1)\in\mathcal{O}(r,1,1,f(r,1),f(r,1)),$$which means that $$d_{L-CSI}(r)=0.$$
\item If $f(r,1)>\frac{1}{2}$, then
\begin{IEEEeqnarray*}{l}
\left(\frac{r}{2(1-f(r,1))},\frac{r}{2(1-f(r,1))}\right)\\
\quad\quad\quad\in\,\mathcal{O}(r,1,1,f(r,1),f(r,1)),
\end{IEEEeqnarray*}which means that
\begin{IEEEeqnarray*}{rCll}
d_{L-CSI}(r) & \leq & 2 - \frac{r}{1-f(r,1)} \\
& \leq & 2-2r, 
\end{IEEEeqnarray*}
since $f(r,1) > \frac{1}{2}.$
\item If $\frac{r}{2} < f(r,1) \leq \frac{1}{2},$ then
$$\left(\frac{r-f(r,1)}{1-f(r,1)},1\right)\in\mathcal{O}(r,1,1,f(r,1),f(r,1)).$$ (Note: Along with checking that the outage condition is satisfied, we also need to check $0\leq\frac{r-f(r,1)}{1-f(r,1)}\leq 1$. This is indeed true since $r < 1$ and $f(r,1)\leq \frac{1}{2} \leq r$.) Hence,
\begin{IEEEeqnarray*}{rCl}
d_{L-CSI}(r) & \leq & 1 - \frac{r-f(r,1)}{1-f(r,1)} \\
& \leq & 1 - \frac{r-\frac{1}{2}}{1-\frac{1}{2}}\\
& = & 2-2r.
\end{IEEEeqnarray*}
\end{itemize}
So by taking the maximum over the 3 cases, we can say that $d_{L-CSI}(r)\leq 2-2r$. 

Achievability by the static QMF scheme with equal listening and transmit times for both the relays follows by \eqref{eq:parstatqmf}. 
\end{proof}

\section{Concluding Remarks}
\label{sec:Conc}

We investigated the necessity of dynamic relaying strategies in achieving the optimal DMT in half-duplex wireless networks with only receive CSI by focusing on the simplest wireless relay network: the single relay channel with arbitrary channel strengths. We introduced a generalized diversity-multiplexing trade-off framework in order to capture different channel strengths in the high-SNR limit. Using this framework, we identified regimes in which dynamic schemes are necessary and those where static schemes are sufficient. We showed that either static QMF or dynamic-decode-forward (DDF) is sufficient to achieve the optimal DMT in all the regimes. Comparing with the full-duplex case, we found that the optimal half-duplex DMT equals the optimal full-duplex DMT in some regimes, while it falls short in certain other regimes. These results also put into perspective earlier results in the literature which focused on the two extreme cases for a single relay network, when all channels are statistically equivalent and when there is no direct link between the source and the destination. 


While static QMF and dynamic-DDF turned out to be sufficient to achieve the optimal DMT of the single relay channel,  we showed through the example of the half-duplex parallel relay network that these two strategies are not sufficient in general to achieve the optimal DMT of larger half-duplex relay networks. We identified a dynamic QMF scheme with a simple listen-transmit schedule for the relays as a function of their receive CSI and proved that it is DMT optimal. For larger networks, there has been significant recent interest in identifying optimal static half-duplex schedules that are globally optimized based on the central knowledge of all the channel realizations in the network \cite{Etkinetal, Brahmaetal,CardoneTK14}. Interestingly, \cite{Brahmaetal,CardoneTK14} show that the optimal schedule uses only a few active states out of the exponentially many possible states for the network. However, such global knowledge of the channel coefficients is rarely available and we believe the DMT framework with local CSI assumption at the relays that we consider in this paper can be more relevant for such networks. Identifying optimal dynamic schedules for more general networks that are functions of the local channel state information available at the relays remains an interesting open problem.

\section*{Acknowledgments}
We would like to thank the reviewers and the associate editor for their insightful comments that improved the exposition in this paper.

\appendices
\section{Proof of Lemma \ref{lem:ddf}}\label{sec:lemddf}
\begin{proof}
We can assume $\gamma\leq\min(\alpha,\beta)$, since assuming $\gamma>\min(\alpha,\beta)$ can be treated identically to the full-duplex case (Lemma \ref{lem:fullduplex}). In DDF, the listening time for the relay is $t =r/\alpha $. Outage occurs if either of the following two events occur:
\begin{itemize}
\item $t > 1$ and $\gamma < r$ (relay never gets a chance to transmit and direct link is not strong enough)
\item $t \leq 1$ and $t\gamma+(1-t)\max(\gamma,\beta) = t\gamma+(1-t)\beta < r$ (relay decodes and is able to transmit but the second cut is not strong enough)
\end{itemize}
The DMT achieved by DDF is given by $d_{DDF}(r) = \min\; a+b+c-\alpha-\beta-\gamma$ where the minimization is over the above two outage events. The first event can be ignored by noting that $\alpha = \min(a,r)$, $\beta = b$ and $\gamma = \min(\alpha,\beta,c,r)=\min(c,r)$ is the optimal point, resulting in $d_{DDF}(r) = (a-r)^+ +(c-r)^+$, which is no worse than the full-duplex DMT. Hence for the remainder of the proof, we deal exclusively with the second event.

From $t=\frac{r}{\alpha}$, we note that $\alpha=\frac{r}{t}$; and since $\alpha\leq a$, we have that $\frac{r}{a}\leq t\leq 1.$ Hence, the domain of outage events is given by the following conditions,
$$ t\gamma+(1-t)\beta < r,$$ $$\frac{r}{a}\leq t\leq 1,$$ $$0\leq \beta\leq b,\; 0\leq \gamma\leq c,$$ $$\gamma\leq\min\left(\frac{r}{t},\beta\right).$$ We instead enlarge the domain of outage events by ignoring the condition to get the following conditions, $\gamma\leq \frac{r}{t}$:
$$t\gamma+(1-t)\beta < r,$$ $$\frac{r}{a}\leq t\leq 1,$$ $$0\leq \beta\leq b,\;0\leq \gamma\leq c,$$ $$\gamma\leq\beta.$$ It can be verified in the end that for every $(a,b,c,r)$ the optimal $(\beta,\gamma,t)$ on the larger domain also satisfy the condition $\gamma\leq \frac{r}{t}$, so enlarging the domain does not affect the solution of the optimization problem.

The optimal values of $\beta,\gamma$ as a function of $(a,b,c,r,t)$ will depend on the relations among $a,b,c,r$ and $t$. A few cases are shown in Figures~\ref{fig:ddf1}, \ref{fig:ddf2}, \ref{fig:ddf3}, \ref{fig:ddf4}, \ref{fig:ddf5}, \ref{fig:ddf6}. The green region denotes the region $\{(\gamma,\beta):0\leq\gamma\leq c,0\leq\beta\leq b, \gamma\leq\beta\}$. The orange region denotes the region $\{(\gamma,\beta):\gamma\geq 0,\beta\geq 0,t\gamma +(1-t)\beta < r\}$ for some $t$. (Note that the boundary of the orange region always passes through $(r,r)$.) The intersection of the two regions, marked by lines, is the feasible region. As can be seen from the different subcases, depending on the slope of the line (which is determined by $t$), the value of $r$, and the relative values of $b$ and $c$, $\beta+\gamma$ is maximized by different points on the boundary of the feasible region. 

\begin{figure}[!htp]
\centering
\begin{tikzpicture}[scale=1.5]
\clip (-0.5,-0.5) rectangle (4,4);
\path[draw,->,>=stealth',name path=x axis] (-0.5,0) -- (3.5,0);
\path[draw,->,>=stealth',name path=y axis] (0,-0.5) -- (0,3.5);
\coordinate (rr) at (0.7,0.7);
\coordinate (c) at (1,0);
\coordinate (b) at (0,2.5);
\coordinate[label=below left:{$\gamma$}] (temp1) at (3.5,0);
\coordinate[label=below left:{$\beta$}] (temp2) at (0,3.5);
\draw let \p1=(c) in (\x1,1pt) -- (\x1,-1pt) node[anchor=north,fill=white] {$c$};
\draw let \p1=(b) in (1pt,\y1) -- (-1pt, \y1) node[anchor=east,fill=white] {$b$};
\filldraw[fill=green!40,fill opacity=0.5] let \p1=(c),\p2=(b) in (0,0) -- (\x1,\x1) -- (\x1,\y2) -- (b) -- cycle;
\draw[dashed] let \p1=(c) in (c) -- (\x1,\x1);
\draw[dashed] let \p1=(c) in (\x1,\x1) -- (2,2);

\coordinate (startpt) at (0,1);
\path let \p1=(startpt),\p2=(rr),\n1={\x2*\y1/(\y1-\y2)} in coordinate (endpt) at (\n1,0){};
\draw (startpt) -- (endpt) node[pos=0.6,anchor=south west]{$t\gamma+(1-t)\beta = r$};
\filldraw[fill=orange!40,fill opacity=0.5] (0,0) -- (startpt) -- (endpt) -- cycle;
\draw[pattern=north east lines] (0,0) -- (rr) -- (startpt) -- cycle;

\node[align=left] (optpttext) at (2,1.5) {optimal};
\draw[->,>=stealth',shorten >=3] (optpttext) -- (rr);
\node[] (temp3) at (0.5,1.3) {$(r,r)$};
\draw[->,>=stealth',shorten >=3] (temp3) -- (rr);
\node at (rr) [circle,draw,fill,inner sep=0pt,minimum size=1mm] {};
\end{tikzpicture}
\caption{$c=1,b=2.5,r<c, t<0.5$}
\label{fig:ddf1}
\end{figure}

\begin{figure}[!htp]
\centering
\begin{tikzpicture}[scale=1.5]
\clip (-0.5,-0.5) rectangle (4,4);
\path[draw,->,>=stealth',name path=x axis] (-0.5,0) -- (3.5,0);
\path[draw,->,>=stealth',name path=y axis] (0,-0.5) -- (0,3.5);
\coordinate (rr) at (0.7,0.7);
\coordinate (c) at (1,0);
\coordinate (b) at (0,2.5);
\coordinate[label=below left:{$\gamma$}] (temp1) at (3.5,0);
\coordinate[label=below left:{$\beta$}] (temp2) at (0,3.5);
\draw let \p1=(c) in (\x1,1pt) -- (\x1,-1pt) node[anchor=north,fill=white] {$c$};
\draw let \p1=(b) in (1pt,\y1) -- (-1pt, \y1) node[anchor=east,fill=white] {$b$};
\filldraw[fill=green!40,fill opacity=0.5] let \p1=(c),\p2=(b) in (0,0) -- (\x1,\x1) -- (\x1,\y2) -- (b) -- cycle;
\draw[dashed] let \p1=(c) in (c) -- (\x1,\x1);
\draw[dashed] let \p1=(c) in (\x1,\x1) -- (2,2);

\coordinate (startpt) at (0,1.6);
\path let \p1=(startpt),\p2=(rr),\n1={\x2*\y1/(\y1-\y2)} in coordinate (endpt) at (\n1,0){};
\draw (startpt) -- (endpt) node[pos=0.9,anchor=south west]{$t\gamma+(1-t)\beta = r$};
\filldraw[fill=orange!40,fill opacity=0.5] (0,0) -- (startpt) -- (endpt) -- cycle;
\draw[pattern=north east lines] (0,0) -- (rr) -- (startpt) -- cycle;

\node[align=left] (optpttext) at (2,1.5) {optimal};
\draw[->,>=stealth',shorten >=3] (optpttext) -- (startpt);
\node[] (temp3) at (0.7,1.3) {$(r,r)$};
\draw[->,>=stealth',shorten >=3] (temp3) -- (rr);
\node at (rr) [circle,draw,fill,inner sep=0pt,minimum size=1mm] {};
\end{tikzpicture}
\caption{$c=1,b=2.5,r<c,t>0.5$}
\label{fig:ddf2}
\end{figure}

\begin{figure}[!htp]
\centering
\begin{tikzpicture}[scale=1.5]
\clip (-0.5,-0.5) rectangle (4,4);
\path[draw,->,>=stealth',name path=x axis] (-0.5,0) -- (3.5,0);
\path[draw,->,>=stealth',name path=y axis] (0,-0.5) -- (0,3.5);
\coordinate (rr) at (0.8,0.8);
\coordinate (c) at (1,0);
\coordinate (b) at (0,1.5);
\coordinate[label=below left:{$\gamma$}] (temp1) at (3.5,0);
\coordinate[label=below left:{$\beta$}] (temp2) at (0,3.5);
\draw let \p1=(c) in (\x1,1pt) -- (\x1,-1pt) node[anchor=north,fill=white] {$c$};
\draw let \p1=(b) in (1pt,\y1) -- (-1pt, \y1) node[anchor=east,fill=white] {$b$};
\filldraw[fill=green!40,fill opacity=0.5] let \p1=(c),\p2=(b) in (0,0) -- (\x1,\x1) -- (\x1,\y2) -- (b) -- cycle;
\draw[dashed] let \p1=(c) in (c) -- (\x1,\x1);
\draw[dashed] let \p1=(c) in (\x1,\x1) -- (2,2);

\coordinate (startpt) at (0,2);
\path let \p1=(startpt),\p2=(rr),\n1={\x2*\y1/(\y1-\y2)} in coordinate (endpt) at (\n1,0){};
\draw[name path=constraint] (startpt) -- (endpt) node[pos=0.9,anchor=south west]{$t\gamma+(1-t)\beta = r$};
\filldraw[fill=orange!40,fill opacity=0.5] (0,0) -- (startpt) -- (endpt) -- cycle;
\path[name path=hor thru b] let \p1=(b) in (b) -- (3,\y1);
\path [name intersections={of=constraint and hor thru b}] coordinate (newpt2) at (intersection-1);
\draw[pattern=north east lines] (0,0) -- (rr) -- (newpt2) -- (b) -- cycle;

\node[align=left] (optpttext) at (2,2.5) {optimal};
\draw[->,>=stealth',shorten >=3] (optpttext) -- (newpt2);
\node[] (temp3) at (0.7,1.3) {$(r,r)$};
\draw[->,>=stealth',shorten >=3] (temp3) -- (rr);
\node at (rr) [circle,draw,fill,inner sep=0pt,minimum size=1mm] {};
\end{tikzpicture}
\caption{$c=1,b=1.5,r<c,t>0.5$}
\label{fig:ddf3}
\end{figure}

\begin{figure}[!htp]
\centering
\begin{tikzpicture}[scale=1.5]
\clip (-0.5,-0.5) rectangle (3.8,4);
\path[draw,->,>=stealth',name path=x axis] (-0.5,0) -- (3.5,0);
\path[draw,->,>=stealth',name path=y axis] (0,-0.5) -- (0,3.5);
\coordinate (rr) at (1.2,1.2);
\coordinate (c) at (1,0);
\coordinate (b) at (0,2.5);
\coordinate[label=below left:{$\gamma$}] (temp1) at (3.5,0);
\coordinate[label=below left:{$\beta$}] (temp2) at (0,3.5);
\draw let \p1=(c) in (\x1,1pt) -- (\x1,-1pt) node[anchor=north,fill=white] {$c$};
\draw let \p1=(b) in (1pt,\y1) -- (-1pt, \y1) node[anchor=east,fill=white] {$b$};
\filldraw[fill=green!40,fill opacity=0.5] let \p1=(c),\p2=(b) in (0,0) -- (\x1,\x1) -- (\x1,\y2) -- (b) -- cycle;
\draw[dashed] let \p1=(c) in (c) -- (\x1,\x1);
\draw[dashed] let \p1=(c) in (\x1,\x1) -- (2,2);

\coordinate (startpt) at (0,1.7);
\path let \p1=(startpt),\p2=(rr),\n1={\x2*\y1/(\y1-\y2)} in coordinate (endpt) at (\n1,0){};
\filldraw[fill=orange!40,fill opacity=0.5] (0,0) -- (startpt) -- (endpt) -- cycle;
\draw [name path=constraint] (startpt) -- (endpt) node[pos=0.7,anchor=north east]{};
\path[name path=vert thru c] let \p1=(c) in (c) -- (\x1,3.5);
\path [name intersections={of=constraint and vert thru c}] coordinate (newpt) at (intersection-1);
\draw[pattern=north east lines] let \p1=(c) in (0,0) -- (\x1,\x1) -- (newpt) -- (startpt) -- cycle;

\node[align=left] (optpttext) at (2,2.5) {optimal};
\draw[->,>=stealth',shorten >=3] (optpttext) -- (newpt);
\node[] (temp3) at (2.5,1.8) {$(r,r)$};
\draw[->,>=stealth',shorten >=3] (temp3) -- (rr);
\node at (rr) [circle,draw,fill,inner sep=0pt,minimum size=1mm] {};
\end{tikzpicture}
\caption{$c=1,b=2.5,r>c,t<0.5$}
\label{fig:ddf4}
\end{figure}

\begin{figure}[!htp]
\centering
\begin{tikzpicture}[scale=1.5]
\clip (-0.5,-0.5) rectangle (4,4);
\draw [->,>=stealth',name path=x axis] (-0.5,0) -- (3.5,0);
\draw [->,>=stealth',name path=y axis] (0,-0.5) -- (0,3.5);
\coordinate (rr) at (1.1,1.1);
\coordinate (c) at (1,0);
\coordinate (b) at (0,2.5);
\coordinate[label=below left:{$\gamma$}] (temp1) at (3.5,0);
\coordinate[label=below left:{$\beta$}] (temp2) at (0,3.5);
\draw let \p1=(c) in (\x1,1pt) -- (\x1,-1pt) node[anchor=north,fill=white] {$c$};
\draw let \p1=(b) in (1pt,\y1) -- (-1pt, \y1) node[anchor=east,fill=white] {$b$};
\filldraw[fill=green!40,fill opacity=0.5] let \p1=(c),\p2=(b) in (0,0) -- (\x1,\x1) -- (\x1,\y2) -- (b) -- cycle;
\draw[dashed] let \p1=(c) in (c) -- (\x1,\x1);
\draw[dashed] let \p1=(c) in (\x1,\x1) -- (2,2);

\coordinate (startpt) at (0,2.4);
\path let \p1=(startpt),\p2=(rr),\n1={\x2*\y1/(\y1-\y2)} in coordinate (endpt) at (\n1,0){};
\draw [name path=constraint] (startpt) -- (endpt) node[pos=0.8,anchor=south west]{$t\gamma+(1-t)\beta=r$};
\filldraw[fill=orange!40,fill opacity=0.5] (0,0) -- (startpt) -- (endpt) -- cycle;
\path[name path=vert thru c] let \p1=(c) in (c) -- (\x1,3.5);
\path [name intersections={of=constraint and vert thru c}] coordinate (newpt) at (intersection-1);
\draw[pattern=north east lines] let \p1=(c) in (0,0) -- (\x1,\x1) -- (newpt) -- (startpt) -- cycle;

\node[align=left] (optpttext) at (2.5,2.3) {optimal};
\draw[->,>=stealth',shorten >=3] (optpttext) -- (startpt);
\node[] (temp3) at (2.5,1.8) {$(r,r)$};
\draw[->,>=stealth',shorten >=3] (temp3) -- (rr);
\node at (rr) [circle,draw,fill,inner sep=0pt,minimum size=1mm] {};
\end{tikzpicture}
\caption{$c=1,b=2.5,r>c,t>0.5$}
\label{fig:ddf5}
\end{figure}

\begin{figure}[!htp]
\centering
\begin{tikzpicture}[scale=1.5]
\clip (-0.5,-0.5) rectangle (4,4);
\draw [->,>=stealth',name path=x axis] (-0.5,0) -- (3.5,0);
\draw [->,>=stealth',name path=y axis] (0,-0.5) -- (0,3.5);
\coordinate (rr) at (1.2,1.2);
\coordinate (c) at (1,0);
\coordinate (b) at (0,2.5);
\coordinate[label=below left:{$\gamma$}] (temp1) at (3.5,0);
\coordinate[label=below left:{$\beta$}] (temp2) at (0,3.5);
\draw let \p1=(c) in (\x1,1pt) -- (\x1,-1pt) node[anchor=north,fill=white] {$c$};
\draw let \p1=(b) in (1pt,\y1) -- (-1pt, \y1) node[anchor=east,fill=white] {$b$};
\filldraw[fill=green!40,fill opacity=0.5] let \p1=(c),\p2=(b) in (0,0) -- (\x1,\x1) -- (\x1,\y2) -- (b) -- cycle;
\draw[dashed] let \p1=(c) in (c) -- (\x1,\x1);
\draw[dashed] let \p1=(c) in (\x1,\x1) -- (2,2);

\coordinate (startpt) at (0,3);
\path let \p1=(startpt),\p2=(rr),\n1={\x2*\y1/(\y1-\y2)} in coordinate (endpt) at (\n1,0){};
\draw [name path=constraint] (startpt) -- (endpt) node[pos=0.8,anchor=south west]{$t\gamma+(1-t)\beta = r$};
\filldraw[fill=orange!40,fill opacity=0.5] (0,0) -- (startpt) -- (endpt) -- cycle;
\path[name path=vert thru c] let \p1=(c) in (c) -- (\x1,3.5);
\path[name path=hor thru b] let \p1=(b) in (b) -- (3,\y1);
\path [name intersections={of=constraint and vert thru c}] coordinate (newpt) at (intersection-1);
\path [name intersections={of=constraint and hor thru b}] coordinate (newpt2) at (intersection-1);
\draw[pattern=north east lines] let \p1=(c) in (0,0) -- (\x1,\x1) -- (newpt) -- (newpt2) -- (b) -- cycle;

\node[align=left] (optpttext) at (2.5,3) {optimal};
\draw[->,>=stealth',shorten >=3] (optpttext) -- (newpt2);
\node[] (temp3) at (2.5,1.8) {$(r,r)$};
\draw[->,>=stealth',shorten >=3] (temp3) -- (rr);
\node at (rr) [circle,draw,fill,inner sep=0pt,minimum size=1mm] {};
\end{tikzpicture}
\caption{$c=1,b=2.5,r>c,t>0.5$}
\label{fig:ddf6}
\end{figure}

We describe the different cases in detail for $c<a<b$. The analysis for $c<b<a$ is similar (though not identical) and hence omitted. The critical values of $t$ are 
\begin{itemize}
\item $\frac{r}{a}$: we know that $t$ has to be greater than this value;
\item $\frac{1}{2}$: when $t=\frac{1}{2}$, the slope of the constraint boundary equals the slope of the objective function;
\item $1-\frac{r}{b}$: when $t=1-\frac{r}{b}$, the $\beta$-coordinate of intersection of constraint region boundary with the $\beta$-axis equals $b$;
\item $\frac{b-r}{b-c}$: when $t=\frac{b-r}{b-c}$, the $\gamma$-coordinate of intersection of constraint boundary with the line $\beta=b$ equals $c.$
\end{itemize}
We need to order these critical values, so that we can identify the optimal $(\beta,\gamma)$ for each $t$. %
For this we note the following:
\begin{itemize}
\item $\frac{r}{a} < \frac{1}{2}$ if $r<\frac{a}{2}$,
\item $\frac{r}{a} < 1-\frac{r}{b}$ if $r<\frac{ab}{a+b}$,
\item $\frac{r}{a} < \frac{b-r}{b-c}$ if $r<\frac{ab}{a+b-c}$,
\item $1-\frac{r}{b}<\frac{1}{2}$ if $r>\frac{b}{2}$,
\item $1-\frac{r}{b}<\frac{b-r}{b-c}$ always,
\item $\frac{b-r}{b-c}<\frac{1}{2}$ if $r>\frac{b+c}{2}$.
\end{itemize}
Finally, for a given $(a,b,c)$ we can calculate the DMT achieved by DDF depending on how the critical values of $r$ identified above are ordered. Tables \ref{tab:dmtddf1} and \ref{tab:dmtddf2} provide all the cases. From the tables, it is easy to verify the claim in Lemma \ref{lem:ddf}.
\end{proof}

\section{Characterizing DMT under global CSI}\label{app:glob}
From the capacity upper-bound and achievable rates described in the preliminaries in Section~\ref{sec:model} for the case of global CSI at the relay, we get the following lower-bound on the probability of outage
\begin{equation}\label{eq:glob_upp}\text{Pr}(\text{outage})\geq\mathbb{P}(\max_{t}C_{h.d.}(\rho^a, \rho^b, \rho^c)+G\leq r\log\rho),\end{equation}
and the following upper-bound on the probability of outage:
\begin{equation}\label{eq:glob_low}\text{Pr}(\text{outage})\leq\mathbb{P}(\max_{t}C_{h.d.}(\rho^a, \rho^b, \rho^c)-\kappa\leq r\log\rho),\end{equation} where $\kappa$ is a constant independent of the SNR.

At high SNR, the constants $G$ and $\kappa$ become insignificant and hence, both \eqref{eq:glob_upp} and \eqref{eq:glob_low} are given by:
\begin{IEEEeqnarray*}{rCl}
\mathbb{P}(\max_{t}C_{h.d.}(\rho^a, \rho^b, \rho^c)\leq r\log\rho) = \mathbb{P}\left(\max_{t(\alpha,\beta,\gamma)}r'_{h.d.}\leq r\right),
\end{IEEEeqnarray*}
where $r'_{h.d.}$ is defined in \eqref{eq:rhdprime} on top of the next page.
\begin{figure*}[!th]
\normalsize
\begin{equation}\label{eq:rhdprime}
r'_{h.d.}\triangleq\min\left\{t\max(\alpha^+,\gamma^+) +(1-t)\gamma^+,t\gamma^++(1-t)\max(\beta^+,\gamma^+)\right\}
\end{equation}
\hrulefill
\end{figure*}

\begin{table*}[!ht]
\normalsize
\centering
\begin{tabular}{l c c c c c c}
\hline
\bfseries &&$(\alpha,\beta,\gamma)$&Optimal $t$& $s(\alpha,\beta,\gamma)$ && $\min s(\alpha,\beta,\gamma)$\\ \hline
$c\leq\frac{a}{2}$\\
\multirow{3}{*}{$\hspace{1 cm}r<c$} & $\frac{r}{a}<t<\frac{1}{2}$ & $(\frac{r}{t},r,r)$ & $\frac{r}{a}$ & $b+c-2r$ & \rdelim\}{3}{1mm}[] & \multirow{3}{*}{$a+c-2r$}\\
& $\frac{1}{2}<t<1-\frac{r}{b}$ & $(\frac{r}{t},\frac{r}{1-t},0)$ & $1-\frac{r}{b}$ & $a+c-\frac{br}{b-r}$ &&\\
& $1-\frac{r}{b}<t<1$ & $(\frac{r}{t},b,b-\frac{b-r}{t})$ & $1$ & $a+c-2r$ &&\vspace{2mm}\\

\multirow{4}{*}{$\hspace{1 cm}c<r<\frac{a}{2}$} & $\frac{r}{a}<t<\frac{1}{2}$ & $(\frac{r}{t},\frac{r-tc}{1-t},c)$ & $\frac{r}{a}$ & $b-\frac{(a-c)r}{a-r}$ & \rdelim\}{4}{1mm}[] & \multirow{4}{*}{$a-\frac{(b-c)r}{b-r}$}\\
& $\frac{1}{2}<t<1-\frac{r}{b}$ & $(\frac{r}{t},\frac{r}{1-t},0)$ & $1-\frac{r}{b}$ & $a+c-\frac{br}{b-r}$ &&\\
& $1-\frac{r}{b}<t<\frac{b-r}{b-c}$ & $(\frac{r}{t},b,b-\frac{b-r}{t})$ & $\frac{b-r}{b-c}$ & $a-\frac{(b-c)r}{b-r}$ &&\\
& $\frac{b-r}{b-c}<t<1$ & $(\frac{r}{t},b,c)$ & $\frac{b-r}{b-c}$ & $a-\frac{(b-c)r}{b-r}$ &&\vspace{2 mm}\\

\multirow{3}{*}{$\hspace{1 cm}\frac{a}{2}<r<\frac{ab}{a+b}$} & $\frac{r}{a}<t<1-\frac{r}{b}$ & $(\frac{r}{t},\frac{r}{1-t},0)$ & $1-\frac{r}{b}$ & $a+c-\frac{br}{b-r}$ & \rdelim\}{3}{1mm}[] & \multirow{3}{*}{$a-\frac{(b-c)r}{b-r}$}\\
& $1-\frac{r}{b}<t<\frac{b-r}{b-c}$ & $(\frac{r}{t},b,b-\frac{b-r}{t})$ & $\frac{b-r}{b-c}$ & $a-\frac{(b-c)r}{b-r}$ &&\\
& $\frac{b-r}{b-c}<t<1$ & $(\frac{r}{t},b,c)$ & $\frac{b-r}{b-c}$ & $a-\frac{(b-c)r}{b-r}$ &&\vspace{2 mm}\\ 

\multirow{2}{*}{$\hspace{1 cm}\frac{ab}{a+b}<r<\frac{b}{2}$} & $\frac{r}{a}<t<\frac{b-r}{b-c}$ & $(\frac{r}{t},b,b-\frac{b-r}{t})$ & $\frac{b-r}{b-c}$ & $a-\frac{(b-c)r}{b-r}$ & \rdelim\}{2}{1mm}[] & \multirow{2}{*}{$a-\frac{(b-c)r}{b-r}$}\\
& $\frac{b-r}{b-c}<t<1$ & $(\frac{r}{t},b,c)$ & $\frac{b-r}{b-c}$ & $a-\frac{(b-c)r}{b-r}$ &&\vspace{2 mm}\\

\multirow{2}{*}{$\hspace{1 cm}\frac{b}{2}<r$} & $\frac{r}{a}<t<\frac{b-r}{b-c}$ & $(\frac{r}{t},b,b-\frac{b-r}{t})$ & $\frac{r}{a}$ & $\frac{ab}{r}-a-b+c$ & \rdelim\}{2}{1mm}[] & \multirow{2}{*}{$\frac{ab}{r}-a-b+c$}\\
& $\frac{b-r}{b-c}<t<1$ & $(\frac{r}{t},b,c)$ & $\frac{b-r}{b-c}$ & $a+c-\frac{br}{b-r}$ &&\\ \hline

$\frac{a}{2}<c\leq\frac{ab}{a+b}$\\
\multirow{3}{*}{$\hspace{1 cm}r<\frac{a}{2}$} & $\frac{r}{a}<t<\frac{1}{2}$ & $(\frac{r}{t},r,r)$ & $\frac{r}{a}$ & $b+c-2r$ & \rdelim\}{3}{1mm}[] & \multirow{3}{*}{$a+c-2r$}\\
& $\frac{1}{2}<t<1-\frac{r}{b}$ & $(\frac{r}{t},\frac{r}{1-t},0)$ & $1-\frac{r}{b}$ & $a+c-\frac{br}{b-r}$ &&\\
& $1-\frac{r}{b}<t<1$ & $(\frac{r}{t},b,b-\frac{b-r}{t})$ & $1$ & $a+c-2r$ &&\vspace{2 mm}\\ 

\multirow{2}{*}{$\hspace{1 cm}\frac{a}{2}<r<c$} & $\frac{r}{a}<t<1-\frac{r}{b}$ & $(\frac{r}{t},\frac{r}{1-t},0)$ & $1-\frac{r}{b}$ & $a+c-\frac{br}{b-r}$ & \rdelim\}{2}{1mm}[] & \multirow{2}{*}{$a+c-2r$}\\
& $1-\frac{r}{b}<t<1$ & $(\frac{r}{t},b,b-\frac{b-r}{t})$ & $1$ & $a+c-2r$ &&\vspace{2 mm}\\

\multirow{3}{*}{$\hspace{1 cm}c<r<\frac{ab}{a+b}$} & $\frac{r}{a}<t<1-\frac{r}{b}$ & $(\frac{r}{t},\frac{r}{1-t},0)$ & $1-\frac{r}{b}$ & $a+c-\frac{br}{b-r}$ & \rdelim\}{3}{1mm}[] & \multirow{3}{*}{$a-\frac{(b-c)r}{b-r}$}\\
& $1-\frac{r}{b}<t<\frac{b-r}{b-c}$ & $(\frac{r}{t},b,b-\frac{b-r}{t})$ & $\frac{b-r}{b-c}$ & $a-\frac{(b-c)r}{b-r}$ &&\\
& $\frac{b-r}{b-c}<t<1$ & $(\frac{r}{t},b,c)$ & $\frac{b-r}{b-c}$ & $a-\frac{(b-c)r}{b-r}$ &&\vspace{2 mm}\\

\multirow{2}{*}{$\hspace{1 cm}\frac{ab}{a+b}<r<\frac{b}{2}$} & $\frac{r}{a}<t<\frac{b-r}{b-c}$ & $(\frac{r}{t},b,b-\frac{b-r}{t})$ & $\frac{b-r}{b-c}$ & $a-\frac{(b-c)r}{b-r}$ & \rdelim\}{2}{1mm}[] & \multirow{2}{*}{$a-\frac{(b-c)r}{b-r}$}\\
& $\frac{b-r}{b-c}<t<1$ & $(\frac{r}{t},b,c)$ & $\frac{b-r}{b-c}$ & $a-\frac{(b-c)r}{b-r}$ &&\vspace{2 mm}\\ 

\multirow{2}{*}{$\hspace{1 cm}\frac{b}{2}<r$} & $\frac{r}{a}<t<\frac{b-r}{b-c}$ & $(\frac{r}{t},b,b-\frac{b-r}{t})$ & $\frac{r}{a}$ & $\frac{ab}{r}-a-b+c$ & \rdelim\}{2}{1mm}[] & \multirow{2}{*}{$\frac{ab}{r}-a-b+c$}\\
& $\frac{b-r}{b-c}<t<1$ & $(\frac{r}{t},b,c)$ & $\frac{b-r}{b-c}$ & $a+c-\frac{br}{b-r}$ &&\\ \hline
\end{tabular}
\caption{DMT achieved by DDF on half-duplex $(a,b,c)$-relay channel, $a<b$, for the cases $c\leq \frac{a}{2}$ and $\frac{a}{2}<c\leq\frac{ab}{a+b}$}
\label{tab:dmtddf1} 
\end{table*}

So the expression for the minimum outage probability under global CSI is given by
\begin{IEEEeqnarray*}{l}
\text{Pr}(\text{outage}) \\
 =  \int p(\alpha)p(\beta)p(\gamma)\cdot\mathbb{P}\left(\max_{t(\alpha,\beta,\gamma)}r'_{h.d.}\leq r\;\Big| \alpha,\beta,\gamma\right)\; d\alpha\;d\beta\;d\gamma
\end{IEEEeqnarray*}
The quantity $\mathbb{P}\left(\max_{t(\alpha,\beta,\gamma)}r'_{h.d.}\leq r\;\Big| \alpha,\beta,\gamma\right)$ has value either 0 or 1. It is 0 if for the given $(\alpha,\beta,\gamma)$, there exists a $t$ such that $r'_{h.d.}>r$. It is 1 if for the given $(\alpha,\beta,\gamma)$, it is the case that $r'_{h.d.}\leq r$ for any choice of t.
Hence,
\begin{IEEEeqnarray}{rCl}
\text{Pr}(\text{outage}) & = & \iiint_{\max_{t(\alpha,\beta,\gamma)}r'_{h.d.}\leq r} p(\alpha)p(\beta)p(\gamma)\; d\alpha\;d\beta\;d\gamma\nonumber\\
& \doteq & \rho^{-d_{G-CSI}(r)},\label{eq:glob_sol}
\end{IEEEeqnarray}
where \begin{equation*}
d_{G-CSI}(r) = \min_{\substack{(\alpha,\beta,\gamma):\\\max_{t(\alpha,\beta,\gamma)} r'_{h.d.}\leq r, \\ \alpha\leq a,\;\beta\leq b,\;\gamma\leq c}}a+b+c-\alpha-\beta-\gamma,\end{equation*} and \eqref{eq:glob_sol} follows by taking the same steps as in \cite[Appendix~A]{ZheTse}. The notation $f(\rho)\doteq g(\rho)$ means that $\lim_{\rho\rightarrow\infty}\frac{\log f(\rho)}{\log\rho}=\lim_{\rho\rightarrow\infty}\frac{\log g(\rho)}{\log\rho}.$


Finally, it is easy to check that introducing non-negativity conditions on $\alpha,\beta,\gamma$ and redefining $r'_{h.d.}$ as \eqref{eq:rhd} instead of \eqref{eq:rhdprime} does not change the problem, which gives us \eqref{eq:globalopt}.\hfill$\QED$

\begin{table*}[!th]
\normalsize
\centering
\begin{tabular}{l c c c c c c}
\hline
\bfseries &&$(\alpha,\beta,\gamma)$&Optimal $t$& $s(\alpha,\beta,\gamma)$ && $\min s(\alpha,\beta,\gamma)$\\ \hline
$\frac{ab}{a+b}<c\leq\frac{b}{2}$\\
\multirow{3}{*}{$\hspace{1 cm}r<\frac{a}{2}$} & $\frac{r}{a}<t<\frac{1}{2}$ & $(\frac{r}{t},r,r)$ & $\frac{r}{a}$ & $b+c-2r$ & \rdelim\}{3}{1mm}[] & \multirow{3}{*}{$a+c-2r$}\\
& $\frac{1}{2}<t<1-\frac{r}{b}$ & $(\frac{r}{t},\frac{r}{1-t},0)$ & $1-\frac{r}{b}$ & $a+c-\frac{br}{b-r}$ &&\\
& $1-\frac{r}{b}<t<1$ & $(\frac{r}{t},b,b-\frac{b-r}{t})$ & $1$ & $a+c-2r$ &&\vspace{2 mm}\\

\multirow{2}{*}{$\hspace{1 cm}\frac{a}{2}<r<\frac{ab}{a+b}$} & $\frac{r}{a}<t<1-\frac{r}{b}$ & $(\frac{r}{t},\frac{r}{1-t},0)$ & $1-\frac{r}{b}$ & $a+c-\frac{br}{b-r}$ & \rdelim\}{2}{1mm}[] & \multirow{2}{*}{$a+c-2r$}\\
& $1-\frac{r}{b}<t<1$ & $(\frac{r}{t},b,b-\frac{b-r}{t})$ & $1$ & $a+c-2r$ &&\vspace{2 mm}\\

\multirow{1}{*}{$\hspace{1 cm}\frac{ab}{a+b}<r<c$} & $\frac{r}{a}<t<1$ & $(\frac{r}{t},b,b-\frac{b-r}{t})$ & $1$ & $a+c-2r$ & \rdelim\}{1}{1mm}[] & \multirow{1}{*}{$a+c-2r$}\vspace{2 mm}\\

\multirow{2}{*}{$\hspace{1 cm}c<r<\frac{b}{2}$} & $\frac{r}{a}<t<\frac{b-r}{b-c}$ & $(\frac{r}{t},b,b-\frac{b-r}{t})$ & $\frac{b-r}{b-c}$ & $a-\frac{(b-c)r}{b-r}$ & \rdelim\}{2}{1mm}[] & \multirow{2}{*}{$a-\frac{(b-c)r}{b-r}$}\\
& $\frac{b-r}{b-c}<t<1$ & $(\frac{r}{t},b,c)$ & $\frac{b-r}{b-c}$ & $a-\frac{(b-c)r}{b-r}$ &&\vspace{2 mm}\\

\multirow{2}{*}{$\hspace{1 cm}\frac{b}{2}<r$} & $\frac{r}{a}<t<\frac{b-r}{b-c}$ & $(\frac{r}{t},b,b-\frac{b-r}{t})$ & $\frac{r}{a}$ & $\frac{ab}{r}-a-b+c$ & \rdelim\}{2}{1mm}[] & \multirow{2}{*}{$\frac{ab}{r}-a-b+c$}\\
& $\frac{b-r}{b-c}<t<1$ & $(\frac{r}{t},b,c)$ & $\frac{b-r}{b-c}$ & $a+c-\frac{br}{b-r}$ &&\\ \hline

$\frac{b}{2}<c$\\
\multirow{3}{*}{$\hspace{1 cm}r<\frac{a}{2}$} & $\frac{r}{a}<t<\frac{1}{2}$ & $(\frac{r}{t},r,r)$ & $\frac{r}{a}$ & $b+c-2r$ & \rdelim\}{3}{1mm}[] & \multirow{3}{*}{$a+c-2r$}\\
& $\frac{1}{2}<t<1-\frac{r}{b}$ & $(\frac{r}{t},\frac{r}{1-t},0)$ & $1-\frac{r}{b}$ & $a+c-\frac{br}{b-r}$ &&\\
& $1-\frac{r}{b}<t<1$ & $(\frac{r}{t},b,b-\frac{b-r}{t})$ & $1$ & $a+c-2r$ &&\vspace{2 mm}\\

\multirow{2}{*}{$\hspace{1 cm}\frac{a}{2}<r<\frac{ab}{a+b}$} & $\frac{r}{a}<t<1-\frac{r}{b}$ & $(\frac{r}{t},\frac{r}{1-t},0)$ & $1-\frac{r}{b}$ & $a+c-\frac{br}{b-r}$ & \rdelim\}{2}{1mm}[] & \multirow{2}{*}{$a+c-2r$}\\
& $1-\frac{r}{b}<t<1$ & $(\frac{r}{t},b,b-\frac{b-r}{t})$ & $1$ & $a+c-2r$ &&\vspace{2 mm}\\

\multirow{1}{*}{$\hspace{1 cm}\frac{ab}{a+b}<r<\frac{b}{2}$} & $\frac{r}{a}<t<1$ & $(\frac{r}{t},b,b-\frac{b-r}{t})$ & $1$ & $a+c-2r$ & \rdelim\}{1}{1mm}[] & \multirow{1}{*}{$a+c-2r$}\vspace{2 mm}\\

\multirow{1}{*}{$\hspace{1 cm}\frac{b}{2}<r<c$} & $\frac{r}{a}<t<1$ & $(\frac{r}{t},b,b-\frac{b-r}{t})$ & $\frac{r}{a}$ & $\frac{ab}{r}-a-b+c$ & \rdelim\}{1}{1mm}[] & \multirow{1}{*}{$\frac{ab}{r}-a-b+c$}\vspace{2 mm}\\

\multirow{2}{*}{$\hspace{1 cm}c<r$} & $\frac{r}{a}<t<\frac{b-r}{b-c}$ & $(\frac{r}{t},b,b-\frac{b-r}{t})$ & $\frac{r}{a}$ & $\frac{ab}{r}-a-b+c$ & \rdelim\}{2}{1mm}[] & \multirow{2}{*}{$\frac{ab}{r}-a-b+c$}\\
& $\frac{b-r}{b-c}<t<1$ & $(\frac{r}{t},b,c)$ & $\frac{b-r}{b-c}$ & $a+c-\frac{br}{b-r}$ &&\\ \hline
\end{tabular}
\caption{DMT achieved by DDF on half-duplex $(a,b,c)$-relay channel, $a<b$, for the cases $\frac{ab}{a+b}<c\leq\frac{b}{2}$ and $\frac{b}{2}<c$}
\label{tab:dmtddf2} 
\end{table*}

\section{Characterizing DMT under local CSI}\label{app:local}
Using the results described in the preliminaries in Section~\ref{sec:model} for the case when the relay only has CSIR, and following similar steps as Appendix~\ref{app:glob}, we get that the optimal probability of outage, which is achievable in the high-SNR limit, is given by
\begin{IEEEeqnarray*}{l}
\text{Pr}(\text{outage})\\  
\quad =  \int_{\alpha} p(\alpha)\min_{t(\alpha)} \text{Pr}\left(r'_{h.d.}\leq r\;\Big| \alpha\right)\; d\alpha \\
\quad =  \int_{\alpha}p(\alpha)  \min_{t(\alpha)} \left(\iint_{(\beta,\gamma)\;:\; r'_{h.d.}\leq r}p(\beta)\; p(\gamma)\; d\beta\; d\gamma \right) d\alpha
\end{IEEEeqnarray*}

Substituting the expressions for $p(\alpha),p(\beta)$ and $p(\gamma),$ and ignoring constants and terms that do not contribute to the $\rho$-exponent, we get that
\begin{IEEEeqnarray*}{l}
\text{Pr}(\text{outage})\doteq \\
 \int_{\alpha\leq a} \rho^{\alpha - a} \min_{t(\alpha)}\left(\iint_{\beta \leq b,\;\gamma\leq c,\; r'_{h.d.}\leq r} \rho^{\beta +\gamma -b-c}\; d\beta \; d\gamma\right) d\alpha.
\end{IEEEeqnarray*}

Now we need to show that $F(\rho)\doteq \rho^{-d_{L-CSI(r)}},$ where $F(\rho)$ is defined to be
\begin{equation}
\int_{\alpha\leq a} \rho^{\alpha - a} \min_{t(\alpha)}\left(\iint_{\beta \leq b,\;\gamma\leq c,\; r'_{h.d.}\leq r} \rho^{\beta +\gamma -b-c}\; d\beta \; d\gamma\right) d\alpha\label{eq:1}
\end{equation}
where \begin{equation}\label{eq:2}d_{L-CSI(r)} = \min_{\alpha\leq a}\max_{t(\alpha)}\min_{\beta\leq b,\;\gamma\leq c,\; r'_{h.d.}\leq r} a+b+c-\alpha-\beta-\gamma.\end{equation}

The proof of this fact follows on similar lines as \cite[Appendix~A]{ZheTse}, the details of which are provided below for the sake of completeness.


\subsection*{Upper bound on $F(\rho)$}

\begin{figure*}[!th]
\normalsize
\begin{IEEEeqnarray*}{ll}
\int_{\alpha< -(b+c)} \rho^{\alpha - a} \min_{t(\alpha)}\left(\iint_{\beta \leq b,\;\gamma\leq c,\; r'_{h.d.}\leq r} \rho^{\beta +\gamma -b-c}\; d\beta \; d\gamma\right) d\alpha \\
 =  \rho^{-(a+b+c)} \int_{\mu\leq 0} \rho^{\mu}\min_{t(\mu)}\left(\iint_{\beta \leq b,\;\gamma\leq c,\; r'_{h.d.}\leq r} \rho^{\beta +\gamma -b-c}\; d\beta \; d\gamma\right) d\mu\\
 \leq  k\rho^{-(a+b+c)}\quad \text{ for some } k<\infty\\
 \doteq  \rho^{-(a+b+c)}
\end{IEEEeqnarray*}
\hrulefill
\end{figure*}

Let $I$ denote the 3-dimensional region $[-(b+c), a] \times [-(a+c), b] \times [-(a+b), c].$ Consider what happens while evaluating $F(\rho)$ if instead of integrating over $\alpha$ in the range $(-\infty,a]$, we integrate over the range $(-\infty,-(b+c)).$ The first chain of inequalities given on the top of the next page shows that this does not change the $\rho$-exponent of $F(\rho)$. In this chain of inequalities, the first step follows by substituting $\mu = b+c+\alpha,$ and the last-but-one step follows since the triple integral is finite.

Hence, for the purpose of evaluating the $\rho$-exponent of $F(\rho)$ we can ignore this term and assume that $\alpha \geq -(b+c).$ Similarly, we can assume $\beta \geq -(a+c)$ and $\gamma \geq -(a+b)$. Then, as shown by the second chain of inequalities on top of the next page, we have 
\begin{equation}\label{eq:ub}
F(\rho) \;\dot{\leq}\; \rho^{-d_{L-CSI(r)}}.
\end{equation}

\begin{figure*}[!th]
\normalsize
\begin{IEEEeqnarray*}{rCl}
F(\rho) &
 \dot{\leq} & (a+b+c)^2 \int_{\alpha< -(b+c)} \rho^{\alpha - a} \,\min_{t(\alpha)}\,
\, \rho^{-\min\limits_{\beta\leq b,\;\gamma\leq c,\;r'_{h.d.}\leq r}\,(b+c-\beta-\gamma)} \,\,d\alpha\\
& = & (a+b+c)^2 \int_{\alpha< -(b+c)} \rho^{\alpha - a} \,\,
\, \rho^{-\max\limits_{t(\alpha)}\,\, \min\limits_{\beta\leq b,\;\gamma\leq c,\;r'_{h.d.}\leq r}\,(b+c-\beta-\gamma)} \,\,d\alpha\\
& \leq & (a+b+c)^3\;\rho^{-d_{L-CSI(r)}}\\
& \doteq  &\rho^{-d_{L-CSI(r)}}
\end{IEEEeqnarray*}
\hrulefill
\end{figure*}

\subsection*{Lower bound on $F(\rho)$}

Define $f(\alpha,\beta,\gamma) \triangleq a+b+c-\alpha-\beta-\gamma$ and
\begin{IEEEeqnarray*}{l}(\alpha^*,\beta^*,\gamma^*) \\
= \arg\;\min_{\alpha\leq a}\max_{t(\alpha)}\min_{\beta\leq b,\;\gamma\leq c,\; r'_{h.d.}\leq r} a+b+c-\alpha-\beta-\gamma.\end{IEEEeqnarray*}
 Since $f(\alpha,\beta,\gamma)$ is continuous, there exists a neighborhood $J$ of $(\alpha^*,\beta^*,\gamma^*)$ within which $f(\alpha,\beta,\gamma) \leq f(\alpha^*,\beta^*,\gamma^*) +\delta$ for any $\delta > 0$. 

Note that around any $(\alpha,\beta,\gamma)$ within the range of integration in \eqref{eq:1}, there exists a neighborhood of points that also lie in the range of integration (e.g. the neighborhood obtained by considering points of the form $(\alpha-\epsilon_1,\beta-\epsilon_2,\gamma-\epsilon_3)$ for $\epsilon_1,\epsilon_2,\epsilon_3>0$). Consider such a neighborhood around $(\alpha^*,\beta^*,\gamma^*)$ and call it $\mathcal{A}.$ Then,
\begin{IEEEeqnarray*}{rCl}
F(\rho) & \geq & \text{vol}(\mathcal{A}\cap J) \rho^{-(f(\alpha^*,\beta^*,\gamma^*)+\delta)}\\
& \doteq & \rho^{-(f(\alpha^*,\beta^*,\gamma^*)+\delta)} \quad\text{ since }\text{vol}(\mathcal{A}\cap J) \neq 0.
\end{IEEEeqnarray*}

Since this is true for all $\delta > 0$, we have a lower bound on $F(\rho)$:
\begin{equation}\label{eq:lb} F(\rho)\;\dot{\geq} \;\rho^{-(a+b+c-\alpha^*-\beta^*-\gamma^*)}.\end{equation}

Hence, from \eqref{eq:ub} and \eqref{eq:lb}, $$F(\rho)\doteq \rho^{-(a+b+c-\alpha^*-\beta^*-\gamma^*)}.$$\hfill$\QED$

\section{Characterizing DMT under no CSI (Static)}\label{app:stat}
The best static scheme chooses $t$ to minimize the probability of outage without the knowledge of any channel realization. As in Appendix~\ref{app:glob}, since the capacity under static schemes and the rate achievable by an appropriate static QMF differ only by constants, the optimal probability of outage, which is achievable in the high-SNR limit, is given by
\begin{IEEEeqnarray}{rCl}
\text{Pr}(\text{outage}) & = & \min_t\left(\iiint_{r'_{h.d.}\leq r} p(\alpha)p(\beta)p(\gamma)\; d\alpha\;d\beta\;d\gamma\right)\nonumber\\
& \doteq & \min_{t}\max_{\substack{r'_{h.d.}\leq r,\;\alpha\leq a,\\ \beta\leq b,\;\gamma\leq c}}\rho^{-(a+b+c-\alpha-\beta-\gamma)}\label{eq:stat_sol}\\
& = &\rho^{-d_{SQMF}(r)}\nonumber,
\end{IEEEeqnarray}
where
\begin{IEEEeqnarray*}{rCl}d_{SQMF}(r) & = & \max_t\;\min_{\substack{r'_{h.d.}\leq r,\; \alpha\leq a,\\\beta\leq b,\;\gamma\leq c}} a+b+c-\alpha-\beta-\gamma,\\
& = & \;\max_t\;\min_{\substack{r_{h.d.}\leq r,\; 0\leq\alpha\leq a,\\ 0\leq\beta\leq b,\; 0\leq\gamma\leq c}} a+b+c-\alpha-\beta-\gamma,
\end{IEEEeqnarray*} and \eqref{eq:stat_sol} is arrived at by following similar steps as \cite[Appendix~A]{ZheTse}.\hfill$\QED$

\bibliographystyle{IEEEtran}
\bibliography{IEEEfull,dmt}

\end{document}